\documentclass[letterpaper,twocolumn,10pt]{article}
\usepackage{usenix,epsfig,endnotes}

\usepackage{libertine}
\usepackage[T1]{fontenc}
\usepackage{notoccite}
\usepackage[numbers,sort]{natbib}
\usepackage{graphicx}
\usepackage{times}
\usepackage{listings}
\usepackage{xspace}
\usepackage[english]{babel}
\usepackage{multirow}
\usepackage{lmodern}
\usepackage{algorithm}
\usepackage{algorithmic}
\usepackage{float}
\usepackage[svgnames]{xcolor}
\usepackage{url}
\usepackage{tikz}
\usepackage{graphicx}

\usepackage{microtype}
\usepackage[TABBOTCAP]{subfigure}
\usepackage{thumbpdf}
\usepackage{listings}
\usepackage{verbatim}

\usepackage{thumbpdf}
\usepackage{listings}
\usepackage{balance}
\usepackage{clrscode}
\usepackage{enumitem}
\usepackage{gensymb}
\usepackage[utf8]{inputenc}
\usepackage{amsthm}
\usepackage{mathtools}
\usepackage{amsmath}
\usepackage{amsopn}
\usepackage{amsopn,amssymb} 
\usepackage{mathtools}
\usepackage[small,compact]{titlesec}

\usepackage{float}
\usepackage[svgnames]{xcolor}
\graphicspath{{Figures/}}

\usepackage{dsfont,bm}

\newcommand{\Section}[1]{\section{#1}}

\date{}

\newtheorem{theorem}{Theorem}
\newtheorem*{theorem*}{Theorem}

\newtheorem{lemma}{Lemma}
\theoremstyle{definition}
\newtheorem{definition}{Definition}
\theoremstyle{definition}
\newtheorem{remark}{Remark}
\newtheorem*{remark*}{Remark}
\newtheoremstyle{assume}
  {3pt}
  {3pt}
  {}
  {}
  {\bf}
  {}
  { }
  {\thmname{#1}.\thmnumber{#2}\thmnote{ \textnormal{(\textit{#3})}}}
\theoremstyle{assume}

\DeclareMathOperator{\DKL}{D_{KL}}

\let\Pr\relax
\DeclareMathOperator{\Pr}{\mathbb{P}}

\DeclareMathOperator*{\argmax}{argmax}
\DeclareMathOperator*{\minimize}{minimize}

\newcommand{\norm}[1]{\ensuremath{\left\| #1 \right\|}}

\newcommand{\calA}{\ensuremath{\mathcal{A}}}

\newcommand{\calC}{\ensuremath{\mathcal{C}}}
\newcommand{\calD}{\ensuremath{\mathcal{D}}}
\newcommand{\calE}{\ensuremath{\mathcal{E}}}
\newcommand{\calF}{\ensuremath{\mathcal{F}}}

\newcommand{\calL}{\ensuremath{\mathcal{L}}}

\newcommand{\calO}{\ensuremath{\mathcal{O}}}

\newcommand{\calT}{\ensuremath{\mathcal{T}}}

\newcommand{\bA}{\ensuremath{\bm{A}}}
\newcommand{\bB}{\ensuremath{\bm{B}}}

\newcommand{\bG}{\ensuremath{\bm{G}}}

\newcommand{\bS}{\ensuremath{\bm{S}}}

\def\st/{\textsuperscript{st}}
\def\nd/{\textsuperscript{nd}}
\def\rd/{\textsuperscript{rd}}
\def\th/{\textsuperscript{th}}

\newcommand{\npod}{\ensuremath{n_\textup{pod}}}
\def\Sys{007\xspace}

\newcommand{\includesvg}[2][scale=1]{%
\includegraphics[#1]{#2.pdf}%
}

\date{}
\title{\vspace{-2.0cm} \Sys: Democratically Finding The Cause of Packet Drops\\(Extended Version)}
\usepackage[affil-sl]{authblk}
\makeatletter
\renewcommand\AB@affilsepx{,\protect\Affilfont}
\renewcommand\Affilfont{\itshape\small}
\makeatother
\author[1]{\rm Behnaz Arzani}
\author[2]{\rm Selim Ciraci}
\author[3]{\rm Luiz Chamon}
\author[1]{\rm Yibo Zhu}
\author[1]{\rm Hongqiang Liu}
\author[2]{\rm Jitu Padhye}
\author[3]{\rm Boon Thau Loo}
\author[2]{\rm Geoff Outhred}
\affil[1]{Microsoft Research}
\affil[2]{Microsoft}
\affil[3]{University of Pennsylvania}

\begin{document}


\maketitle

\noindent{\textbf{Abstract --}} Network failures continue to plague datacenter operators as
        their symptoms may not have direct correlation with where or why they occur. We introduce \Sys, a lightweight, always-on
        diagnosis application that can find problematic links and also pinpoint problems {\em for each
        TCP connection}. \Sys is completely contained within the end host. During its two month deployment in a tier-1 datacenter,  it detected every problem found by previously deployed monitoring tools while also finding the sources of other problems previously undetected.  

\vspace{-1mm}
\section{Introduction}
\label{sec:introduction}
\vspace{-1mm}

\Sys has an ambitious goal: for every packet drop on a TCP flow in a
datacenter, find the link that dropped the packet and do so
with negligible overhead and no changes to the network infrastructure.

This goal may sound like an overkill---after all, TCP is supposed to be able to
deal with a few packet losses. Moreover, packet losses might occur due to
congestion instead of network equipment failures. Even network failures might
be transient. Above all, there is a danger of drowning in a sea of data without
generating any actionable intelligence.

These objections are valid, but so is the need to diagnose
``failures'' that can result in severe problems for
applications. For example, in our datacenters, VM images are stored
in a storage service. When a VM boots, the image is mounted
over the network. Even a small network outage or a few lossy links can cause the VM to ``panic'' and reboot. In fact, 17\% of our VM
reboots are due to network issues and in over $70\%$ of these none of our monitoring tools were able to find the 
links that caused the problem.




VM reboots affect customers and we need to
understand their root cause.  Any persistent pattern in such transient
failures is a cause for concern and is potentially actionable. One example is silent packet drops~\cite{pingmesh}.
These types of problems are nearly impossible to detect with traditional monitoring
tools (e.g., SNMP).  If a switch is experiencing these problems, we may
want to reboot or replace it.  These interventions are
``costly'' as they affect a large number of flows/VMs. Therefore, careful blame assignment is necessary. Naturally, this is only one example that would benefit from such a detection system.

There is a lot of prior work on network failure diagnosis, though one of the existing systems meet our ambitious goal. Pingmesh~\cite{pingmesh}
sends periodic probes to detect failures and can leave ``gaps'' in coverage, as it must manage the overhead of probing.
Also, since it uses out-of-band probes, it 
cannot detect failures that affect only in-band data. Roy et al.~\cite{alex} monitor all paths to detect failures but require modifications to routers and special features in the switch~(\S\ref{sec:relatedwork}).
Everflow~\cite{everflow} can be used to find the
location of packet drops but it would require capturing all traffic and is
not scalable. We asked our operators what would be the most useful solution for them. Responses included: ``In a network of $\ge 10^6$ links its a reasonable assumption that there is a non-zero chance that
a number ($>10)$ of these links are bad (due to device, port, or cable, etc.) and we cannot fix them simultaneously. Therefore, fixes need to be prioritized based on
customer impact. However, currently we do not have a direct way to correlate customer impact with bad links". This shows that current systems do not satisfy operator needs as they do not provide application and connection level context.

To address these limitations, we propose \Sys, a simple, lightweight,
always-on monitoring tool. \Sys records the path of TCP connections~(flows)
suffering from one or more retransmissions and assigns proportional ``blame''
to each link on the path. It then provides a {\em ranking} of links
that represents their relative drop rates. Using this ranking, it can find the most likely cause
of drops in each TCP flow.

\Sys has several noteworthy properties. First, it does not require any changes to the existing networking infrastructure. Second, it does not require changes to the client software---the monitoring agent is an independent entity that
sits on the side. Third, it detects in-band failures. Fourth, it continues to perform well in the presence of noise (e.g.~lone packet drops). Finally, it's overhead is negligible.

While the high-level design of \Sys appear simple, the
practical challenges of making \Sys work and the theoretical challenge
of {\em proving} it works are non-trivial.  For example, its
path discovery is based on a traceroute-like approach. Due to the
use of ECMP, traceroute packets have to be carefully crafted to ensure
that they follow the same path as the TCP flow. Also, we
must ensure that we do not overwhelm routers 
by sending too many traceroutes (traceroute responses are handled by
control-plane CPUs of routers, which are quite puny). Thus,
we need to ensure that our sampling strikes
the right balance between accuracy and the overhead on
the switches.  On the theoretical side, we are able to show that \Sys's
simple blame assignment scheme is highly accurate even in the presence
of noise.

We make the following contributions: (i)~we design \Sys, a simple, lightweight, and yet accurate fault localization system for datacenter networks; (ii)~we prove that \Sys is accurate without imposing excessive burden on the switches; (iii)~we prove that its blame assignment scheme correctly finds the failed links with high probability; and (iv)~we show how to tackle numerous practical challenges involved in deploying \Sys in a real datacenter. 

Our results from a two month deployment of \Sys in a datacenter show that it finds
all problems found by other previously deployed monitoring tools while
also finding the sources of problems for which information is not provided
by these monitoring tools.  


\vspace{-1mm}
\section{Motivation}
\label{sec:motivation}

\Sys aims to identify the cause of retransmissions with high probability. It is is driven by two practical requirements: (i)~it should scale to datacenter size networks and (ii)~it should be deployable in a running datacenter with as little change to the infrastructure as possible. Our current focus is mainly on analyzing infrastructure traffic, especially connections to services such as storage as these can have severe consequences~(see~\S\ref{sec:introduction},~\cite{Netpoirot}).  Nevertheless, the same mechanisms can be used in other contexts as well~(see~\S\ref{sec:discussion}). We deliberately include congestion-induced retransmissions. If episodes of congestion, however short-lived, are common on a link, we want to be able to flag them. Of course, in practice, any such system needs to deal with a certain amount of noise, a concept we formalize later.

There are a number of ways to find the cause of packet drops. One can
monitor switch counters. These are inherently
unreliable~\cite{netpilot} and monitoring thousands of switches at a
fine time granularity is not scalable. One can use new hardware
capabilities to gather more useful
information~\cite{mohammad}. Correlating this data with each
retransmission {\em reliably} is difficult. Furthermore, time is
needed until such hardware is production-ready and switches are
upgraded. Complicating matters, operators may be
unwilling to incur the expense and overhead of such
changes~\cite{Netpoirot}. One can use PingMesh~\cite{pingmesh} to send
probe packets and monitor link status. Such systems suffer from a rate
of probing trade-off: sending too many probes creates unacceptable
overhead whereas reducing the probing rate leaves temporal and spatial
gaps in coverage. More importantly, the probe traffic does not capture
what the end user and TCP flows see. Instead, we
choose to use data traffic itself as probe traffic. Using data traffic
has the advantage that the system introduces little to no monitoring
overhead.

As one might expect, almost all traffic in our datacenters is TCP
traffic. One way to monitor TCP traffic is to use a system like
Everflow. Everflow inserts a special tag in every packet and has the
switches mirror tagged packets to special collection servers. Thus, if
a tagged packet is dropped, we can determine the link on which it
happened. Unfortunately, there is no way to know in advance which
packet is going to be dropped, so we would have to tag and mirror
every TCP packet. This is clearly infeasible. We could tag only a
fraction of packets, but doing so would result in another sampling
rate trade-off. Hence, we choose to rely on some form of network
tomography~\cite{tomo1,tomo2,tomo3}. We can take advantage of the fact
that TCP is a connection-oriented, reliable delivery protocol so that
any packet loss results in retransmissions that are easy to detect.

If we knew the path of all flows, we could set up an optimization to
find which link dropped the packet. Such an optimization would minimize the number of ``blamed'' links while simultaneously explaining the cause of all
drops. Indeed past approaches such as MAX COVERAGE
and Tomo~\cite{netdiagnoser, risk} aim to approximate the solution of such an optimization (see~\S\ref{sec:appendix} for an example).
There are
problems with this approach: (i)~the optimization is
NP-hard~\cite{bertsimas}. Solving it on a datacenter scale is
infeasible. (ii)~tracking the path of every flow in the
datacenter is not scalable in our setting. We can use alternative
solutions such as Everflow or the approach of~\cite{alex} to
track the path of SYN packets. However, both rely on making changes to
the switches. The only way to find the path of a flow without any
special infrastructure support is to employ something like a
traceroute. Traceroute relies on getting ICMP TTL exceeded messages
back from the switches. These messages are generated by the
control-plane, i.e., the switch CPU. To avoid overloading the CPU, our
administrators have capped the rate of ICMP responses to 100 per
second. This severely limits the number of flows we can track.


\begin{figure}[tb]
	\centering
	\includesvg{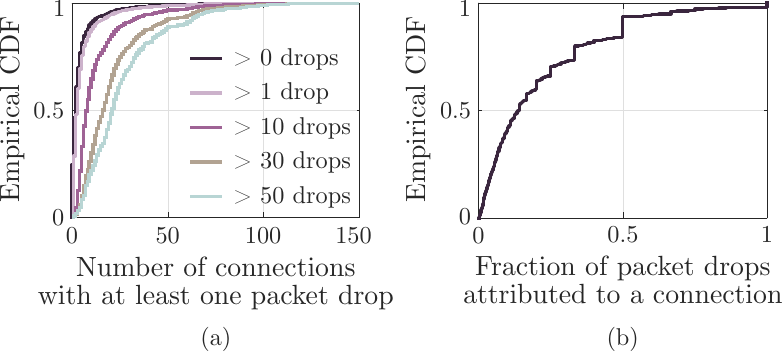}
	\vspace*{-7mm}
	\caption{Observations from a production network: (a)~CDF of the number of flows with at least one retransmission; (b)~CDF of the fraction of drops belonging to each flow in each $30$ second interval.}
	\label{fig:numflowswithdrops}
	\vspace{-7mm}
\end{figure}

Given these limitations, what can we do? We analyzed the drop patterns in two of our datacenters and found:
typically when there are packet drops, multiple flows
experience drops. We show this in Figure~\ref{fig:numflowswithdrops}a
for TCP flows in production datacenters. The figure shows the number of flows
experiencing drops in the datacenter conditioned on the total number
of packets dropped in that datacenter in $30$ second intervals. The
data spans one day.  We see that the
more packets are dropped in the datacenter, the more flows experience
drops and $95\%$ of the time, at least $3$ flows see
drops when we condition on $\ge 10$ total drops. We focus
on the $\ge 10$ case because lower values mostly capture noisy drops
due to one-off packet drops by healthy links.  In most cases
drops are distributed across flows and no single flow sees more
than $40\%$ of the total packet drops.  This is shown in
Figure \ref{fig:numflowswithdrops}b (we have discarded all flows with
$0$ drops and cases where the total number of drops was less than
10). We see that in $\ge 80\%$ of cases, no single flow captures more
than $34\%$ of all drops.

Based on these observations and the high path diversity in datacenter
networks~\cite{f10}, we show that if: (a)~we only track the path of those flows that have retransmissions, (b)~assign each link on the
path of such a flow a vote of $1/h$, where $h$ is the path length, and
(c)~sum up the votes during a given period, then the top-voted links
are almost always the ones dropping packets~(see \S\ref{sec:vote})!
Unlike the optimization, our scheme is able to provide a {\em ranking}
of the links in terms of their drop rates, i.e. if link $A$ has a
higher vote than $B$, it is also dropping more packets~(with high
probability). This gives us a heat-map of our network which highlights
the links with the most impact to a {\em given application/customer}
(because we know which links impact a particular flows). 


\vspace{-1mm}
\section{Design Overview}
\vspace{-1mm}

\begin{figure}[t]
  \centering
    \includesvg{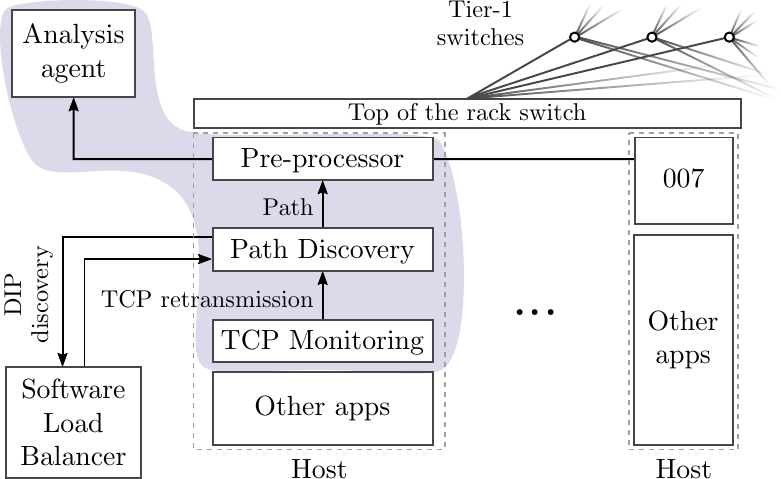}
    \caption{Overview of \Sys architecture\label{fig:overview}}
	\vspace{-6mm}
\end{figure}


Figure~\ref{fig:overview} shows the overall architecture of
\Sys. It is deployed alongside other
applications on each end-host as a user-level process running in the
host OS.  \Sys consists of three agents responsible for TCP monitoring, path discovery, and analysis.

The {\em TCP monitoring agent} detects
retransmissions at each end-host. The Event Tracing For Windows
(ETW)~\cite{etw} framework\footnote{Similar functionality exists in
  Linux.} notifies the agent as soon as an active flow
suffers a retransmission.

Upon a retransmission, the monitoring agent triggers the {\em
  path discovery agent}~(\S\ref{sec:pathdiscovery}) which identifies the flow's path to the destination IP (DIP).

At the end-hosts, a voting scheme~(\S\ref{sec:vote}) is used based on the paths of flows that had retransmissions. At regular intervals of $30$s the votes are tallied by a centralized {\em analysis agent} to find the top-voted links. Although we use an aggregation interval of $30$s, failures \emph{do not} have to last for $30$s.

\Sys's implementation consists of 6000 lines of C++ code. Its memory usage never goes beyond $600$~KB on any
of our production hosts, its CPU utilization is minimal~($1$-$3\%$), and its bandwidth utilization due to traceroute is minimal~(maximum of $200$~KBps per host). \Sys is proven to be accurate (\S\ref{sec:vote}) in typical datacenter conditions (a full description of the assumed conditions can be found in~\S\ref{sec:discussion}).


\vspace{-1mm}
\section{The Path Discovery Agent} 
\label{sec:pathdiscovery} 
\vspace{-1mm}

The path discovery agent uses traceroute packets to find the path of flows that suffer retransmissions. These packets are used solely to identify the path of a flow. They do not need to be dropped for \Sys to operate. We first ensure that the number of traceroutes sent by the agent does not overload our switches~(\S\ref{sec:cpu}). Then, we briefly describe the key engineering issues and how we solve them~(\S\ref{sec:eng}). 

\vspace{-1mm}
\subsection{ICMP Rate Limiting}
\label{sec:cpu}
\vspace{-1mm}
Generating ICMP packets in response to traceroute consumes switch CPU,
which is a valuable resource.  In our network, there
is a cap of $T_\text{max}=100$ on the number of ICMP messages a switch can send per second. To ensure that the traceroute load does not exceed $T_\text{max}$, we start by noticing that a small fraction of flows go through tier-3 switches~($T_3$). Indeed, after monitoring all TCP flows in our network for one hour, only~$2.1\%$ went through a~$T_3$ switch. Thus we can ignore $T_3$ switches in our analysis. Given
that our network is a Clos topology and assuming that hosts
under a top of the rack switch~(ToR) communicate with hosts under a
different ToR uniformly at random (see~\S\ref{sec:sim} for when this is not the case):


\vspace{-1mm}
\begin{theorem}\label{theorem1}
The rate of ICMP~packets sent by any switch
due to a traceroute is below~$T_\textup{max}$ if the rate~$C_t$ at which
hosts send traceroutes is upper bounded as
\begin{align}\label{eq:maxC}
    C_t \leq \frac{
         T_\textup{max} } {n_0 H} \min \left[ n_1, \frac{n_2(n_0 \npod-1)} {n_0(\npod-1)}   \right] 
        \text{,}
\end{align}
where~$n_0$, $n_1$, and~$n_2$, are the numbers of ToR, $T_1$, and~$T_2$ switches respectively, $n_{pod}$ is the number of pods, and~$H$ is the number of hosts under each ToR.

\end{theorem}

%

See~\S\ref{sec:appendix} for proof.
The upper bound of~$C_t$ in our datacenters is~$10$. As long as hosts do not have more than~$10$ flows \emph{with retransmissions} per second, we can \emph{guarantee} that the number of traceroutes sent by \Sys will not go above $T_\text{max}$. We use $C_t$ as a threshold to limit the traceroute rate of each host. Note that there are two independent rate limits, one set at the host by \Sys and the other set by the network operators on the switch~($T_\text{max}$). Additionally, the agent triggers path discovery \emph{for a given connection} no more than once every epoch to further limit the number of traceroutes.  We will show in \S\ref{sec:vote} that this number is sufficient to ensure high accuracy.

\vspace{-1mm}
\subsection{Engineering Challenges}
\label{sec:eng}
\vspace{-1mm}
\noindent{\bf Using the correct five-tuple.} As in most datacenters, our network also uses ECMP. All packets of a given flow, defined by the five-tuple, follow the same path~\cite{ECMP}. Thus, traceroute packets must have the same five-tuple as the flow we want to trace. To ensure this, we must account for load balancers.

TCP connections are initiated in our datacenter in a way similar
to that described in~\cite{ananta}. The connection is first
established to a virtual IP (VIP) and the SYN packet (containing the VIP
as destination) goes to a software load balancer (SLB) which
assigns that flow to a physical destination IP (DIP) and a
service port associated with that VIP. The SLB then sends a
configuration message to the virtual switch (vSwitch) in the
hypervisor of the source machine that registers that DIP with that
vSwitch. The destination of all subsequent packets in that flow
have the DIP as their destination and do not go through the SLB. For the path of the traceroute packets to match that of the data
packets, its header should contain the DIP and not the VIP. Thus,
before tracing the path of a flow, the path discovery agent
first queries the SLB for the VIP-to-DIP mapping for that flow. An alternative is to query the vSwitch. In the instances where the failure also results in connection termination the mapping may be removed from the vSwitch table. It is therefore more reliable to query the SLB. Note that there are cases where the TCP connection
establishment itself may fail due to packet loss. Path discovery is
not triggered for such connections. It is also not triggered when the
query to the SLB fails to avoid tracerouting the internet.

\noindent{\bf Re-routing and packet drops.} Traceroute itself may fail. This may happen if the link drop rate is high or due to a blackhole.  This actually helps us, as it directly pinpoints the faulty link and our analysis engine~(\S\ref{sec:vote}) is able to use such partial traceroutes.

A more insidious possibility is that routing may change by the time
traceroute starts. We use BGP in our datacenter and a lossy link
may cause one or more BGP sessions to fail, triggering rerouting.
Then, the traceroute packets may take a different path than the original
connection. However, RTTs in a datacenter are
typically less than 1 or 2~ms, so TCP retransmits a
dropped packet quickly. The ETW framework notifies the monitoring agent
immediately, which invokes the path discovery agent. The only additional delay
is the time required to query the SLB to obtain the VIP-to-DIP mapping, which is
typically less than a millisecond. Thus, as long as paths are stable for a few
milliseconds after a packet drop, the traceroute packets will follow the same
path as the flow and the probability of error is low. Past work has shown this to be usually the case~\cite{ffc}.

Our network also makes use of link aggregation~(LAG)~\cite{linkagg}.
However, unless all the links in the aggregation group fail, the L3 path is not
affected.

\noindent{\bf Router aliasing~\cite{aliasing}.} This problem is easily solved in a datacenter, as we
know the topology, names, and IPs of all routers and interfaces. We can
simply map the IPs from the traceroutes to the switch names.

To summarize, \Sys's path discovery implementation is as follows:
Once the TCP monitoring agent notifies the path discovery agent that a
flow has suffered a retransmission, the path discovery agent checks its
cache of discovered path for that epoch and if need be, queries the SLB for the
DIP.  It then sends~$15$ appropriately crafted TCP packets with TTL values ranging
from~$0$--$15$. In order to disambiguate the responses, the TTL value is also
encoded in the IP ID field~\cite{ip}. This allows for concurrent
traceroutes to multiple destinations.  The TCP packets deliberately carry a
bad checksum so that they do not interfere with the ongoing connection.


\vspace{-1mm}
\Section{The Analysis Agent}
\label{sec:vote}
\vspace{-1mm}



Here, we describe \Sys's analysis agent focusing on its voting-based scheme. We also present alternative NP-hard optimization solutions for comparison.

\vspace{-1mm}
\subsection{Voting-Based Scheme}
    \label{sec:votingdetails}
\vspace{-1mm}

\Sys's analysis agent uses a simple voting scheme. If a
flow sees a retransmission, \Sys votes its links as \emph{bad}. Each vote has a value that is tallied at the end of every epoch, providing a natural ranking of the links. We set the value of good votes to~$0$ (if a flow has no retransmission, no traceroute
is needed). Bad votes are assigned a value of~$\frac{1}{h}$,
where~$h$ is the number of hops on the path, since each link on
the path is equally likely to be responsible for the drop.

The ranking obtained after compiling the votes allows us to
identify the most likely cause of drops on each flow: links ranked higher have higher drop rates~(Theorem~\ref{T:vigilWorks}). To further
guard against high levels of noise, we can use our knowledge of the
topology to adjust the links votes. Namely, we iteratively pick the most voted link~$l_\text{max}$ and estimate the portion of votes obtained by all other links due to failures on~$l_\text{max}$. This estimate is obtained for each link~$k$ by (i)~assuming all flows having retransmissions and going through~$l_\text{max}$ had drops due to~$l_\text{max}$ and (ii)~finding what fraction of these flows go through~$k$ by assuming ECMP distributes flows uniformly at random. Our evaluations showed that this results in a $5\%$ reduction in false positives.

\begin{algorithm}
 \begin{algorithmic}[1]
 {\scriptsize
\STATE $\mathcal{L} \gets$  \text{Set of all links}\\
\STATE $\mathcal{P} \gets$ \text {Set of all {\em possible} paths}\\
\STATE $v(l_i) \gets$ \text{Number of votes for $l_i \in \mathcal{L}$}\\
\STATE $\mathcal{B} \gets$ \text{Set of most problematic links}\\
\STATE $l_{max} \gets \text{Link with maximum votes in $\forall l_i \in \mathcal{L} \cap \mathcal{B}^c$}$\\
\WHILE{$ v(l_{max}) \ge 0.01(\sum_{l_i \in \mathcal{L}} v(l_i))$}
  \STATE $ l_{max} \gets \argmax_{l_i \in \mathcal{L} \cap \mathcal{B}^c} v(l_i)$
  \STATE $\mathcal{B} \gets \mathcal{B} \cup \{l_{max}\}$
  \FOR {$l_i \in  \mathcal{L} \cap \mathcal{B}^c$}
  \IF{ $\exists \quad p_i \in \mathcal{P} \quad \text{s.t. } l_i \in p_i  \text{ \& }  l_{max} \in p_i$}
  \STATE Adjust the score of $l_i$
  \ENDIF
  \ENDFOR
\ENDWHILE
\RETURN  $\mathcal{B}$
}

\end{algorithmic}

\caption{Finding the most problematic links in the network.\label{alg:vigil}}

\end{algorithm}

\Sys can also be used to detect failed links using Algorithm~\ref{alg:vigil}.
The algorithm sorts the links based on their votes and uses a threshold to determine if there are problematic links. If so, it adjusts the votes of all other links and repeats until no link has votes above the threshold. In Algorithm~\ref{alg:vigil}, we use a threshold of~$1\%$ of the total votes cast based on a parameter sweep where we found that it provides a reasonable trade-off between precision and recall. Higher values reduce false positives but increase false negatives.

Here we have focused on detecting 
link failures. \Sys can also
be used to detect switch failures in a similar fashion by applying votes to switches instead of links. This is beyond the scope of
this work.

\vspace{-1mm}
\subsection{Voting Scheme Analysis}

\label{sec:reliable}
\vspace{-1mm}

Can \Sys deliver on its promise of finding the most probable cause of
packet drops on each flow? This is not trivial. In its voting scheme,
failed connections contribute to increase the tally of both good and bad
links. Moreover, in a large datacenter such as ours, occasional,
lone, and sporadic drops can and \emph{will} happen due to good
links. These failures are akin to noise and can cause severe
inaccuracies in any detection system~\cite{gestalt}, \Sys
included. We show that the likelihood of \Sys making
these errors is small. Given our topology~(Clos):

\vspace{-1mm}
\begin{theorem}
\label{T:vigilWorks}

For~$\npod \geq \frac{n_0}{n_1} + 1$, \Sys will find with
probability~$1-2e^{-\calO(N)}$ the~$k < \frac{n_2 (n_0 \npod - 1)}{n_0 (\npod - 1)}$ bad links
that drop packets with probability~$p_b$ among good links that drop packets with probability $p_g$ if
\begin{align*}
	p_g \leq (n_{u} \alpha)^{-1} \left[ 1-(1-p_b)^{n_{l}} \right], 
\end{align*}
where $N$ is the total number of flows between hosts, $n_l$ and~$n_u$ are
lower and upper bounds, respectively, on the number of packets per
connection, and
\begin{equation}\label{eq:alpha}
    \alpha =
    \frac{
        n_0 (4 n_0 - k) (\npod - 1)
    }{
        n_2 (n_0 \npod - 1) - n_0 (\npod - 1) k
    }
        \text{.}
\end{equation}

\end{theorem}

%
%


The proof is deferred to the appendices due to space constraints. Theorem~\ref{T:vigilWorks} states that under mild conditions, links with higher drop rates are ranked higher by \Sys. Since a single flow is unlikely to go through more than one failed link in a network with thousands of links, it allows \Sys to find the most likely cause of packet drops on each flow.

A corollary of Theorem~\ref{T:vigilWorks} is
that in the absence of noise~($p_g = 0$), \Sys can find all bad
links with high probability. In the presence of
noise, \Sys can still identify the bad links as long as the
probability of dropping packets on non-failed links is low
enough~(the signal-to-noise ratio is large enough). This number is compatible with typical values found in practice. As an example, let~$n_l$ and~$n_u$ be the~$10^{th}$ and~$90^{th}$ percentiles respectively of the number of packets sent by TCP
flows across all hosts in a $3$~hour period. If~$p_b \geq 0.05\%$, the drop rate on good links can be as high as~$1.8 \times 10^{-6}$. Drop rates in a production datacenter are typically below~$10^{-8}$~\cite{RAIL}.

Another important consequence of Theorem~\ref{T:vigilWorks} is that it
establishes that the probability of errors in \Sys's results diminishes exponentially
with $N$, so that even with the limits imposed by
Theorem~\ref{eq:maxC} we can accurately identify the failed links. The conditions in
Theorem~\ref{T:vigilWorks} are sufficient but not necessary. In fact,
\S\ref{sec:sim} shows how well \Sys performs even when the
conditions in Theorem~\ref{T:vigilWorks} do not hold.

\vspace{-1mm}
\subsection{Optimization-Based Solutions}
\label{sec:optimal}

\vspace{-1mm}
One of the advantages of \Sys's voting scheme is its simplicity. Given additional time and resources we may consider searching for the optimal sets of failed links by finding the most likely cause of drops given the available evidence. For instance, we can find the {\em least number of links} that explain all failures as we know the flows that had packet drops and their path. This can be written as an optimization problem we call the~\emph{binary program}. Explicitly,
\vspace{-1mm}
\begin{equation}\label{opt:binary}
\begin{aligned}
	\minimize&
	&&\norm{\bm{p}}_0
	\\
	\text{subject to}&
	&&\bm{A} \bm{p} \geq \bm{s}
	\\
	&&&\bm{p} \in \{0,1\}^L
\end{aligned}
\end{equation}
%
where~$\bm{A}$ is a~$C \times L$ routing matrix; $\bm{s}$ is a~$C \times 1$ vector that collects the status of each flow during an epoch (each element of~$\bm{s}$ is~$1$ if the connection experienced at least one retransmission and~$0$ otherwise); $L$ is the number of links; $C$ is the number of connections in an epoch; and~$\norm{\bm{p}}_0$ denotes the number of nonzero entries of the vector~$\bm{p}$. Indeed, if the solution of~\eqref{opt:binary} is~$\bm{p}^\star$, then the~$i$-th element of~$\bm{p}^\star$ indicates whether the binary program estimates that link~$i$ failed. 

Problem~\eqref{opt:binary} is the NP-hard minimum set covering problem~\cite{combinatorial} and is intractable. Its solutions can be approximated greedily as in MAX COVERAGE or Tomo~ \cite{netdiagnoser, risk} (see appendix). For benchmarking, we compare \Sys to the true solution of~\eqref{opt:binary} obtained by a mixed-integer linear program~(MILP) solver~\cite{mosek}. Our evaluations showed that \Sys~(Algorithm~\ref{alg:vigil}) significantly outperforms this binary optimization~(by more than $50\%$ in the presence of noise). We illustrate this point in Figures~\ref{fig:reasonable} and~\ref{fig:alg1Single}, but otherwise omit results for this optimization in~\S\ref{sec:sim} for clarity.

The binary program~\eqref{opt:binary} does not provide a ranking of links. We also consider a solution in which we determine the number of packets dropped by each link, thus creating a natural ranking. The \emph{integer program} can be written as
\vspace{-1mm}
\begin{equation}\label{opt:integer}
\begin{aligned}
	\minimize&
	&&\norm{\bm{p}}_0
	\\
	\text{subject to}&
	&&\bm{A} \bm{p} \geq \bm{c}
	\\
	&&&\norm{\bm{p}}_1 = \norm{\bm{c}}_1
	\\
	&&&p_i \in \mathbb{N} \cup \{0\}
\end{aligned}
\end{equation}
%
where $\mathbb{N}$ is the set of natural numbers and~$\bm{c}$ is a~$C \times 1$ vector that collects the number of retransmissions suffered by each flow during an epoch. The solution~$\bm{p}^\star$ of~\eqref{opt:integer} represents the number of packets dropped by each link, which provides a ranking. The constraint~$\norm{\bm{p}}_1 = \norm{\bm{c}}_1$ ensures each failure is explained only once. As with~\eqref{opt:binary}, this problem is NP-hard~\cite{bertsimas} and is only used as a benchmark. As it uses more information than the binary program~(the number of failures), \eqref{opt:integer} performs better~(see~\S\ref{sec:sim}).

In the next three sections, we present our evaluation of \Sys in simulations~(\S\ref{sec:sim}), in a test cluster~(\S\ref{sec:test}), and in one of our production datacenters~(\S\ref{sec:prod}).


\vspace{-1mm}
\section{Evaluations: Simulations}
\label{sec:sim}
\vspace{-1mm}

We start by evaluating in simulations where we know the ground
truth. \Sys first finds flows whose drops
were due to noise and marks them as ``noise drops''. It then finds the
link most likely responsible for drops on the remaining set of
flows~(``failure drops''). A noisy drop is defined as one where
the corresponding link only dropped a single packet. \Sys never marked a connection into the noisy category
incorrectly. We therefore focus on the accuracy for connections that
\Sys puts into the failure drop class.

\noindent{\bf Performance metrics.} Our measure for the
performance of \Sys is \emph{accuracy}, which is the proportion of
correctly identified drop causes.
For evaluating Algorithm~\ref{alg:vigil}, we use {\em recall} and {\em
  precision}. Recall is a measure of reliability and shows how many of
the failures \Sys can detect~(false negatives). For
example, if there are $100$ failed links and \Sys detects $90$ of them,
its recall is~$90\%$. Precision is a measure of accuracy and shows to
what extent \Sys's results can be trusted~(false positives). For
example, if \Sys flags $100$ links as bad, but only $90$ of those links
actually failed, its precision is~$90\%$.

\noindent{\bf Simulation setup.} We use a flow level simulator~\cite{sims}
implemented in MATLAB. Our topology consists of $4160$ links, $2$ pods, and~$20$ ToRs per pod. Each host
establishes $2$ connections per second to a random ToR
outside of its rack. The simulator has two types of links. For {\em
  good links}, packets are dropped
at a very low rate chosen uniformly from $(0,10^{-6})$ to simulate
noise. On the other hand, {\em failed links} have a higher drop
rate to simulate failures. By default, drop rates on failed links are
set to vary uniformly from $0.01$\% to $1$\%, though to study the
impact of drop rates we do allow this rate to vary as an input
parameter. The number of good and failed links is also tunable. Every $30$ seconds of simulation time, we send
up to $100$ packets per flow and drop them based on the
rates above as they traverse links along the path. The simulator
records all flows with at least one drop and for each such
flow, the link with the most drops.

We compare \Sys against the solutions described in~\S\ref{sec:optimal}. We only show results for the binary program~\eqref{opt:binary} in Figures~\ref{fig:reasonable} and~\ref{fig:alg1Single} since its performance is typically inferior to \Sys and the integer program~\eqref{opt:integer} due to noise. This also applies to MAX COVERAGE or Tomo~ \cite{netdiagnoser, risk,rachit} \textit{as they are approximations of the binary program}~(see~\cite{tr}).


\vspace{-1mm}
\subsection{In The Optimal Case}
\vspace{-1mm}



The bounds of Theorem~\ref{T:vigilWorks}  are
sufficient (not necessary) conditions for accuracy. We first
validate that \Sys can achieve high levels of accuracy as
expected when these bounds hold. We set the drop rates on
the failed links to be between $(0.05\%,1\%)$. We refer the reader to~\cite{alex} for why these drop rates are reasonable.


\noindent{\bf Accuracy.} Figure~\ref{fig:theorem2} shows that \Sys
has an average accuracy that is higher than $96\%$ in almost all
cases. Due to its robustness to noise, it also
outperforms the optimization algorithm~(\S~\ref{sec:optimal}) in most
cases.


\noindent{\bf Recall \& precision.} Figure~\ref{fig:reasonable}
shows that even when failed links have low packet drop rates, \Sys
detects them with high recall/precision. 


We proceed to evaluate \Sys's accuracy
when the bounds in Theorem~\ref{T:vigilWorks} {\em do not} hold. This shows these conditions are not necessary for good performance.

\begin{figure}[tb]
	\centering
	\includesvg{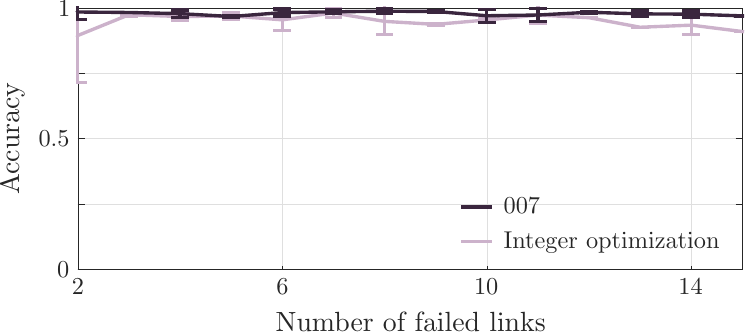}
	\vspace{-2mm}
	\caption{When Theorem 2 holds.}
		\label{fig:theorem2}
		\vspace{-3mm}
\end{figure}

\begin{figure}[tb]
	\centering
	\includesvg{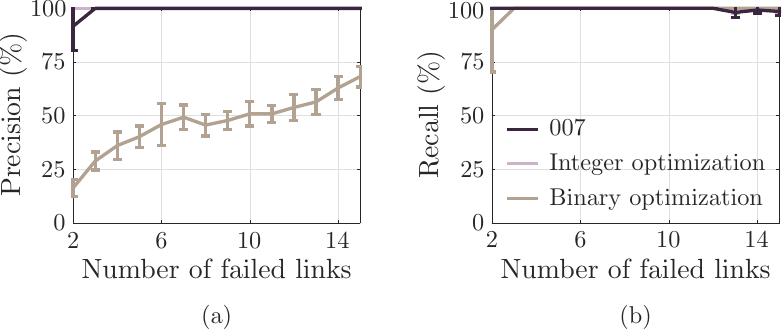}
	\vspace*{-6mm}
	\caption{Algorithm~\ref{alg:vigil} when Theorem 2 holds.}
		\label{fig:reasonable}
	\vspace{-3mm}
\end{figure}

\vspace{-1mm}
\subsection{Varying Drop Rates}
\label{sec:Drops}
\vspace{-1mm}
Our next experiment aims to push the boundaries of
Theorem~\ref{T:vigilWorks} by varying the ``failed'' links drop rates
below the conservative bounds of Theorem~\ref{T:vigilWorks}.


\noindent{\textbf{Single Failure.}} Figure~\ref{fig:dropRate}a shows
 results for different drop rates on a single failed link. It shows that \Sys can find
the cause of drops on each flow with high accuracy. Even
as the drop rate decreases below the bounds of
Theorem~\ref{T:vigilWorks}, we see that \Sys can
maintain accuracy on par with the
optimization.


\noindent{\textbf{Multiple Failures.}} Figure~\ref{fig:dropRate}b shows that \Sys is successful at 
finding the link responsible for a drop even when links have very different drop rates. 
Prior work have reported the difficulty of detecting such cases~\cite{alex}. However, \Sys's accuracy remains high.

\begin{figure}[t]
	\centering
	\includesvg{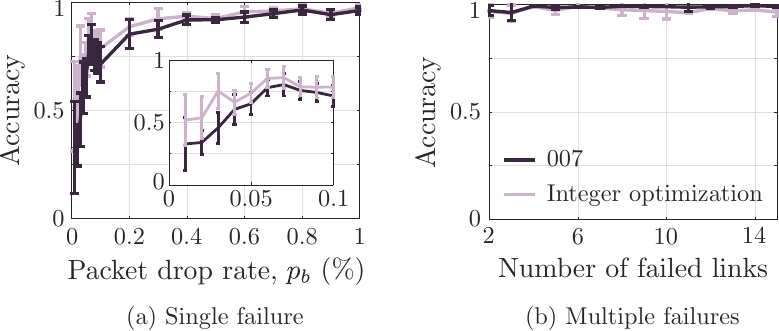}
	\vspace*{-6mm}
	\caption{\Sys's accuracy for varying drop rates.}
		\label{fig:dropRate}
		\vspace{-3mm}
\end{figure}

\vspace{-1mm}
\subsection{Impact of Noise}
	\label{sec:Noise}

\vspace{-1mm}
\noindent{\textbf{Single Failure.}} We vary noise levels by changing the drop rate of good links.
We see that higher noise levels have little impact on \Sys's ability to find the cause of
drops on individual flows~(Figure~\ref{fig:noise}a).


\noindent{\textbf{Multiple Failures.}} We repeat this experiment for
the case of $5$ failed links. Figure~\ref{fig:noise}b shows the
results. \Sys shows little sensitivity to the increase in noise
when finding the cause of per-flow drops. Note that
the large confidence intervals of the optimization
is a result of its high sensitivity to noise.

\begin{figure}[tb]
  \centering
	\includesvg{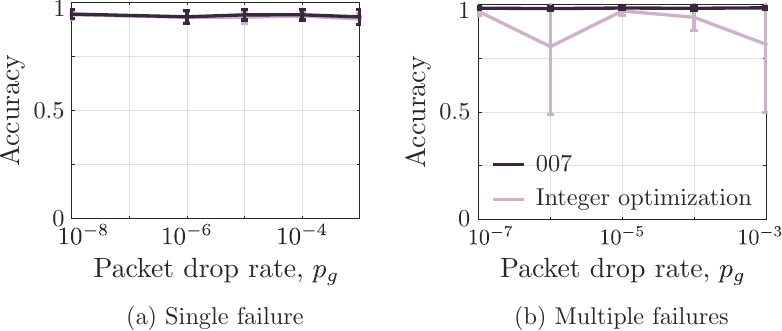}
	\vspace*{-6.5mm}
	\caption{\Sys's accuracy for varying noise levels.}
		\label{fig:noise}
	\vspace{-3mm}
\end{figure}


\vspace{-1mm}
\subsection{Varying Number of Connections}
\label{sec:singleNumCon}
\vspace{-1mm}

In previous experiments, hosts
opened~$60$ connections per epoch. Here, we allow hosts to choose the
number of connections they create per epoch uniformly at random
between $(10,60)$. Recall, from Theorem~\ref{T:vigilWorks}, that a
larger number of connections from each host helps \Sys improve its
accuracy. 

\noindent{\textbf{Single Failure.}}  Figure~\ref{fig:connections}a
shows the results. \Sys accurately finds the cause of  packet drops
on each connection. It also
outperforms the optimization when the failed link has a low
drop rate. This is because the optimization has multiple optimal
points and is not sufficiently constrained.


\noindent{\textbf{Multiple Failures.}} Figure~\ref{fig:connections}b
shows the results for multiple failures.  The optimization
suffers from the lack of information to constrain the set of
results. It therefore has a large variance~(confidence intervals). \Sys on the other hand maintains high
probability of detection no matter the number of
failures.

\begin{figure}[tb]
  \centering
\includesvg{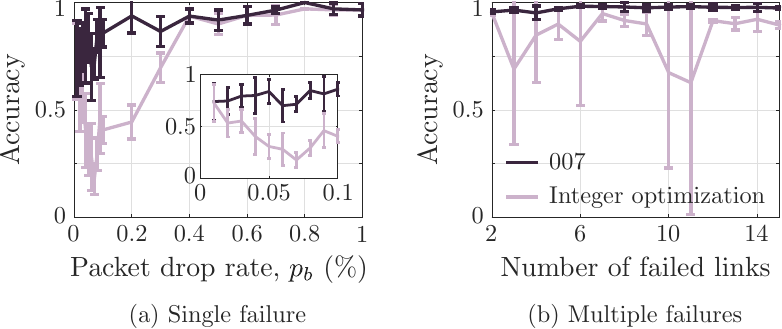}
\vspace*{-7mm}
\caption{Varying  the number of connections.}
		\label{fig:connections}
\end{figure}

\vspace{-1mm}
\subsection{Impact of Traffic Skews}
	\label{sec:singleSkew}
\vspace{-1mm}	

\noindent{\textbf{Single Failure.}} We next demonstrate \Sys's ability
to detect the cause of drops even under heavily skewed traffic.
We pick $10$ ToRs at random ($25\%$ of the ToRs). To skew the traffic, 80\% of the
flows have destinations set to hosts under these $10$ ToRs.
The remaining flows are routed to randomly chosen hosts. Figure~\ref{fig:trafficSkew}a shows that the optimization is much more heavily impacted by the skew than
\Sys. \Sys continues to detect the cause of drops with high
probability ($\ge 85\%$) for drop rates higher than $0.1\%$. 

\begin{figure}[tb]
  	\centering
	\includesvg{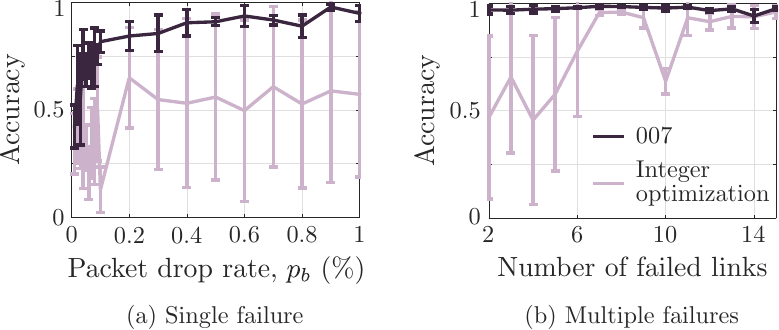}
	\vspace*{-7mm}
     \caption{\Sys's accuracy under skewed traffic.}
		\label{fig:trafficSkew}
	\vspace{-3mm}
\end{figure}

\noindent{\textbf{Multiple Failures.}} We repeated the above for multiple failures. Figure~\ref{fig:trafficSkew}b shows
that the optimization's accuracy suffers. It consistently shows a low
detection rate as its constraints are not sufficient in guiding the
optimizer to the right solution. \Sys maintains a
detection rate of $\ge 98\%$ at all times.

\noindent{\textbf{Hot ToR.}} A special instance of traffic skew occurs in the presence of a single hot ToR which acts as a sink for a large number
of flows. Figure~\ref{fig:hottor} shows how \Sys performs in these situations. \Sys can 
tolerate up to $50\%$ skew, i.e., $50\%$ of {\em all} flows go to the hot ToR, with negligible accuracy degradation. However, skews above $50\%$ negatively impact its
accuracy in the presence of a large number of failures ($\ge10$).  Such scenarios are unlikely as datacenter load balancing mitigates such extreme situations.

\begin{figure}[tb]
	\centering
	\includesvg{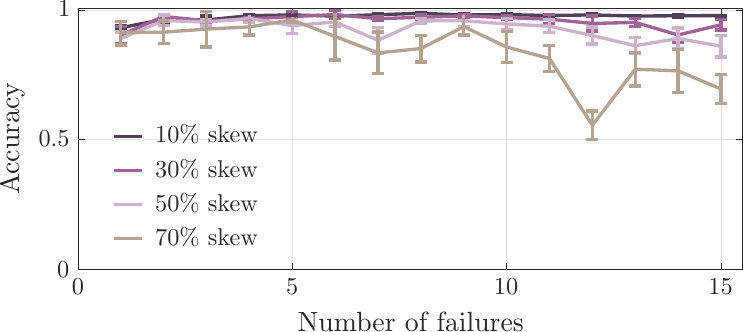}
	\vspace*{-2mm}
    \caption{Impact of a hot ToR on \Sys's accuracy. }
		\label{fig:hottor}
	\vspace{-3mm}
\end{figure}

\vspace{-1mm}
\subsection{Detecting Bad Links}
\label{sec:singleAlg}
\vspace{-1mm}

In our previous experiments, we focused on \Sys's 
accuracy on a per connection basis. In our next experiment,
we evaluate its ability to detect bad links.

\noindent{\textbf{Single Failure.}} Figure~\ref{fig:alg1Single} shows
the results. \Sys outperforms the optimization as it does not require
a fully specified set of equations to provide a best guess as
to which links failed. We also evaluate the impact of failure location 
on our results (Figure~\ref{fig:linklocation}).

\begin{figure}[tb]
  \centering
	\includesvg{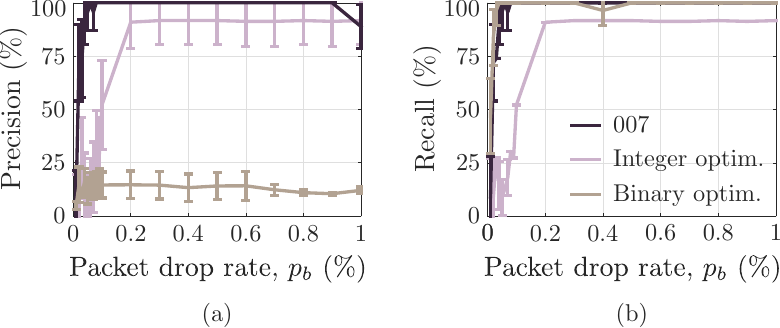}
	\vspace*{-7mm}
     \caption{Algorithm~\ref{alg:vigil} with single failure.}
     \label{fig:alg1Single}
    \vspace{-5mm}
\end{figure}

\noindent{\textbf{Multiple Failures.}} We heavily skew the drop rates on the failed links. Specifically, at
least one failed link has a drop rate between~$10$ and~$100\%$,
while all others have a drop rate in~$(0.01\%,0.1\%)$. This scenario is one
that past approaches have reported as hard to detect~\cite{alex}.
Figure~\ref{fig:alg1Multiple} shows that \Sys can detect up to $7$
failures with accuracy above $90\%$. Its recall drops
as the number of failed links increase. This is because
the increase in the number of failures drives up the votes of all
other links increasing the cutoff threshold and thus increasing
the likelihood of false negatives. In fact if the top $k$
links had been selected \Sys's recall would have been close to
$100\%$.

\begin{figure}[tb]
    \centering
	\includesvg{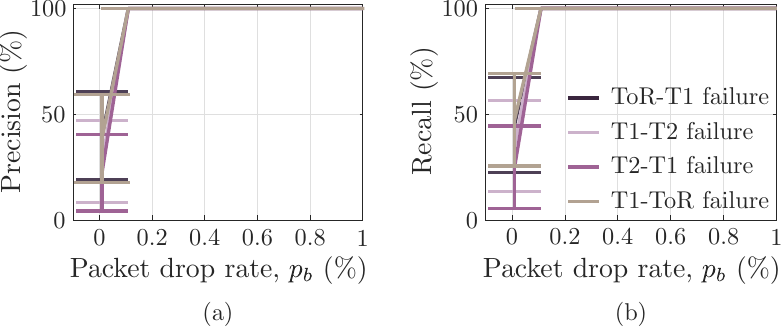}
		\vspace*{-7mm}
    \caption{Impact of link location on Algorithm~\ref{alg:vigil}.}
		\label{fig:linklocation}
	\vspace{-3mm}
\end{figure}

\begin{figure}[tb]
  \centering
\includesvg{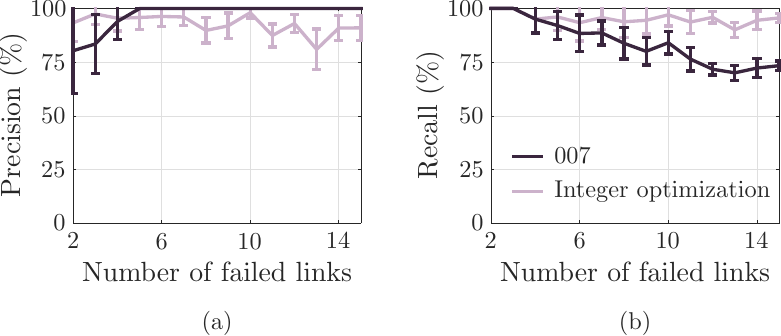}
	\vspace*{-7mm}
   \caption{Algorithm~\ref{alg:vigil} with multiple failures. The drop rates on the links
   are heavily skewed. }
		\label{fig:alg1Multiple}
		\vspace{-5mm}
\end{figure}

\vspace{-1mm}
\subsection{Effects of Network Size}
\vspace{-1mm}

Finally, we evaluate \Sys in larger networks. Its accuracy when finding a single failure was~$98\%$, $92\%$, $91\%$, and $90\%$ on average in a network with $1,2,3,$ and $4$ pods respectively. In
contrast, the optimization had an average
accuracy of~$94\%$, $72\%$, $79\%$, and $77\%$ respectively. Algorithm~\ref{alg:vigil} continues to have Recall $\ge 98\%$ for up to $6$ pods (it drops to $85\%$ for $7$ pods). Precision remains~$100\%$ for all pod sizes.

We also evaluate both \Sys and the optimization's ability to find
the cause of per flow drops when the number of failed links is~$\ge 30$. We observe that both approach's performance
remained unchanged for the most part, e.g., the accuracy of \Sys in an example with~$30$ failed links is
$98.01\%$.

\vspace{-1mm}
\section{Evaluations: Test Cluster}
\label{sec:test}
\vspace{-1mm}


We next evaluate \Sys on the more realistic environment of a test cluster with $10$ ToRs and a total of~$80$ links. We control~$50$ hosts in the cluster, while others are production machines. Therefore, the~$T_1$ switches see real production traffic. We recorded~$6$ hours of traffic from a host in production and replayed it from our hosts in the cluster~(with different starting times). Using Everflow-like functionality~\cite{everflow} on the ToR switches, we induced different rates of drops on~$T_1$ to ToR links. Our goal is to find the cause of packet drops on each flow~\S\ref{sec:testPerCon} and to validate whether Algorithm 1 works in practice~\S\ref{sec:alg1}.



\vspace{-1mm}
\subsection{Clean Testbed Validation}
\vspace{-1mm}

We first validate a clean testbed environment. We repave the cluster by setting all devices to a clean state. We then run \Sys without injecting any failures. We see that in the newly-repaved cluster, links arriving at a particular ToR switch had abnormally high votes, namely $22.5 \pm 3.65$ in average. We thus suspected that this ToR is experiencing problems. After rebooting it, the total votes of the links went down to $0$, validating our suspicions. This exercise also provides one example of when \Sys is extremely effective at identifying links with low drop rates.



\vspace{-1mm}
\subsection{Per-connection Failure Analysis}
\label{sec:testPerCon}
\vspace{-1mm}
Can \Sys identify the cause of drops when links have very different drop rates? To find out, we induce a drop rate of~$0.2\%$ and~$0.05\%$ on two different links for an hour. We only know the ground truth when the flow goes through at least one of the two failed links. Thus, we only consider such flows. For $90.47\%$ of these, \Sys was able to attribute the packet drop to the correct link~(the one with higher drop rate).

\begin{figure}[htb]
  \centering
    \includesvg{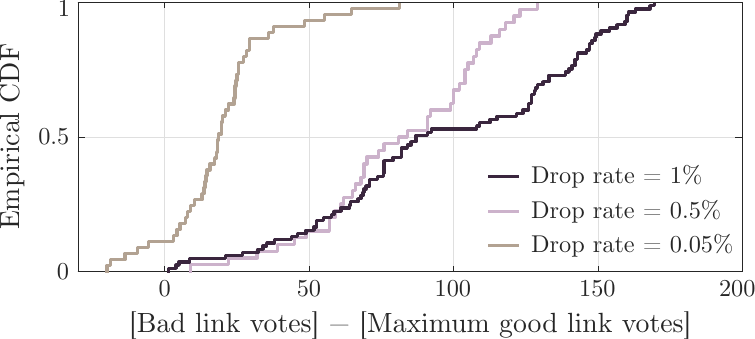}
  	\vspace*{-2mm}
    \caption{Distribution of the difference between votes on bad links and the maximum vote on good links for different bad link drop rates.\label{fig:testclusterdroprates}}
	\vspace{-13mm}
\end{figure}

\vspace{-1mm}
\subsection{Identifying Failed Links}
\label{sec:alg1}
\vspace{-1mm}


We next validate Algorithm~1 and its ability to detect failed links. We inject different drop rates on a chosen link and determine whether there is a correlation between total votes and drop rates. Specifically, we look at the difference between the vote tally on the bad link and that of the most voted good link. We induced a packet drop rate of~$1\%$, $0.1\%$, and~$0.05\%$ on a~$T_1$ to ToR link in the test cluster.

Figure~\ref{fig:testclusterdroprates} shows the distribution for the various drop rates. The failed link has the highest vote out of all links when the drop rate is~$1\%$ and~$0.1\%$. When the drop rate is lowered to $0.05\%$, the failed link becomes harder to detect due to the smaller gap between the drop rate of the bad link and that of the normal links. Indeed, the bad link only has the maximum score in~$88.89\%$ of the instances~(mostly due to occasional lone drops on healthy links). However, it is always one of the~$2$ links with the highest votes.

Figure~\ref{fig:testclusterdroprates} also shows the high correlation between the probability of packet drop on a links and its vote tally. This trivially shows that \Sys is $100\%$ accurate in finding the cause of packet drops on each flow given a single link failure: the failed link has the highest votes among all links. We compare \Sys with the optimization problem in~\eqref{opt:integer}. We find that the latter also returns the correct result every time, albeit at the cost of a large number of false positives. To illustrate this point: the number of links marked as bad by~\eqref{opt:integer} on average is $1.5$, $1.18$, and~$1.47$ times higher than the number given by \Sys for the drop rates of~$1\%$, $0.1\%$, and~$0.05\%$ respectively.

What about multiple failures? This is a harder experiment to configure due to the smaller number of links in this test cluster and its lower path diversity. We induce different drop rates~($p_1 = 0.2\%$ and~$p_2 = 0.1\%$) on two links in the cluster. The link with higher drop rate is the most voted~$100\%$ of the time. The second link is the second highest ranked~$47\%$ of the time and the third~$32\%$ of the time. It always remained among the~$5$ most voted links. This shows that by allowing a single false positive~(identifying three instead of two links), \Sys can detect all failed links~$80\%$ of the time even in a setup where the traffic distribution is highly skewed. This is something past approaches~\cite{alex} could not achieve. In this example, \Sys identifies the true cause of packet drops on each connection~$98\%$ of the time.

\vspace{-2mm}
\section{Evaluations: Production}
\label{sec:prod}
\vspace{-1mm}


We have deployed \Sys in one of our datacenters\footnote{The monitoring agent has been deployed across all our data centers for over $2$ years.}. 
Notable examples of problems \Sys found include: power supply undervoltages~\cite{EOS}, FCS errors~\cite{monia}, switch reconfigurations, continuous BGP state changes, link flaps, and software bugs~\cite{cisco}. Also, \Sys found every problem that was caught by our previously deployed diagnosis tools. 
\vspace{-1mm}
\subsection{Validating Theorem~\ref{theorem1}}
\vspace{-1mm}
Table~\ref{fig:numIcmp} shows the distribution of the number of ICMP messages sent by each switch in each epoch over a week. The number of ICMP messages generated by \Sys never exceed $T_{max}$ (Theorem~\ref{theorem1}). 

\begin{table}
\begin{center}
{
\begin{tabular}{cccc}

\textbf{$T=0$} & \textbf{$T>0$ \& $T\le 3$} & \textbf{$T > 3$} & \textbf{$\max (T)$}\\
\hline 
\hline
$69\%$ & $30.98\%$ & $0.02\%$ & $11$\\
\end{tabular}
}
\caption{Number of ICMPs per second per switch~($T$). We see $\max(T) \le T_{\text{max}}$.}
\label{fig:numIcmp}
\end{center}
\vspace{-9mm}
\end{table}

%
%

\vspace{-1mm}
\subsection{TCP Connection Diagnosis}
\vspace{-1mm}
In addition to finding problematic links, \Sys identifies the most likely cause of drops on each flow.
Knowing when each individual packet is dropped in production is hard. We perform a semi-controlled experiment to test the accuracy of \Sys. Our environment consists of thousands of hosts/links. To find the ``ground truth'',
we compare its results to that obtained by EverFlow. EverFlow captures all packets going
through each switch on which it was enabled. It is expensive to run for extended periods of time. We thus only run EverFlow for $5$ hours and configure it to capture all outgoing IP traffic from $9$ random hosts. The captures for each host were conducted on different days. We filter all flows that were detected to have at least one retransmission during this time and using EverFlow find where their packets were dropped.
We then check whether the detected link matches that found by \Sys. We found that {\em \Sys was accurate in every single case}. In this test we also verified that each path recorded by \Sys matches exactly the path taken by that flow's packets as captured by EverFlow. This confirms that it is unlikely for paths to change fast enough to cause errors in \Sys's path discovery.

\vspace{-1mm}
\subsection{VM Reboot Diagnosis}
\vspace{-1mm}

During our deployment, there were~$281$ VM reboots in the datacenter
for which there was no explanation. \Sys found a link as the cause of
problems in each case. Upon further investigation on the SNMP system
logs, we observe that in~$262$ cases, there were transient drops on
the host to ToR link a number of which were correlated with high CPU
usage on the host. Two were due to high drop rates on the ToR.  In
another~$15$, the endpoints of the links found were undergoing
configuration updates. In the remaining~$2$ instances, the link was
flapping.

Finally, we looked at our data for one cluster for one day. \Sys identifies an average of $0.45\pm0.12$ links as dropping packets per epoch. The average across all epochs of the maximum vote tally was $2.51\pm0.33$. Out of the links dropping packets $48\%$ are server to ToR links ($38\%$ were due to a single ToR switch that was eventually taken out for repair), $24\%$ are $T_1$-ToR links and $6\%$ were due to $T_2$-$T_1$ link failures.


\vspace{-1mm}
\section{Discussion}
\label{sec:discussion}
\vspace{-1mm}

\Sys is highly effective in finding the cause of packet drops on individual flows. By doing so, it provides flow-level context which is useful in finding the cause of problems for specific applications. In this section we discuss a number of other factors we considered in its design.

\subsection{\Sys's Main Assumptions}
The proofs of Theorems~\ref{theorem1} and~\ref{T:vigilWorks} and the design of the path discovery agent (\S\ref{sec:pathdiscovery}) are based on a number of assumptions:

\noindent{\bf ACK loss on reverse path.} It is possible
that packet loss on the reverse path is so severe that loss of ACK
packets triggers timeout at the sender. If this happens, the traceroute
would not be going over any link that triggered the packet drop. Since
TCP ACKs are cumulative, this is typically not a problem and \Sys assumes retransmissions in such cases are unlikely. This is true unless loss
rates are very high, in which case the severity of the problem is such that the cause is apparent. 
Spurious retransmissions triggered by timeouts
may also occur if there is sudden increased delay on forward or reverse
paths. This can happen due to
rerouting, or large queue buildups. \Sys treats these retransmissions like any other.

\noindent{\textbf{Source NATs.}} Source network address translators
(SNATs) change the source IP of a packet before it is sent out to a VIP. Our current implementation of \Sys assumes connections are SNAT bypassed. However, if flows are SNATed, the ICMP messages will not have the right source address for \Sys to get the response to its traceroutes. This can be fixed by a query to the SLB. Details are omitted. 

\noindent{\textbf{L2 networks.}} Traceroute is not a viable option to find paths when datacenters operate using L2 routing. In such cases we recommend one of the following: (a) If access to the destination is not a problem and switches can be upgraded one can use the path discovery methods of~\cite{pathdump,alex}.
\Sys is still useful as it allows for finding the cause of failures when multiple failures are present and for individual flows. (b) Alternatively, EverFlow can be used to find path. \Sys's sampling is necessary here as EverFlow doesn't scale to capture the path of all flows.

\noindent{\textbf{Network topology.}} The
calculations in~\S\ref{sec:vote} assume a known topology~(Clos). The
same calculations can be carried out for {\em any} known
topology by updating the values used for ECMP. The accuracy of \Sys is tied to the degree of path
diversity and that multiple paths are available at each
hop: the higher the degree of path diversity, the better \Sys
performs. This is also a desired property in any datacenter
topology, most of which follow the Clos topology~\cite{googletopo,microsofttopo,f10}.


\noindent{\textbf{ICMP rate limit.}} In rare instances, the severity of a failure or the number of flows impacted by it may be such that it triggers \Sys's ICMP rate limit which stops sending more traceroute messages in that epoch. This {\em does not} impact the accuracy of Algorithm~1. By the time \Sys reaches its rate limit, it has enough data to localize the problematic links. However, this limits \Sys's ability to find the cause of drops on flows for which it did not identify the path. We accept this trade-off in accuracy for the simplicity and lower overhead of \Sys.

\noindent{\textbf{Unpredictability of ECMP.}} If the topology and the ECMP functions on all the routers are known, the path of a packet can be found by inspecting its header. However, ECMP functions are typically proprietary and have
initialization ``seeds'' that change with every reboot of the
switch. Furthermore, ECMP functions change after link failures and recoveries. Tracking all link failures/recoveries in real time is not feasible at a datacenter scale.

\subsection{Other Factors To Consider}
\Sys has been designed for a specific use case, namely finding the cause of packet drops on individual connections in order to provide application context. This resulted in a number of design choices:

\noindent{\textbf{Detecting congestion.}} \Sys~{\em should not} avoid detecting
major congestion events as they signal severe traffic imbalance and/or incast and are actionable. However, the more prevalent ($\ge 92\%$) forms of congestion have low drop rates $10^{-8}$--$10^{-5}$~\cite{monia}. \Sys treats these as noise and does not detect them. Standard congestion control protocols can effectively react to them.

\noindent{\textbf{\Sys's ranking.}}  \Sys's ranking approach will
naturally bias towards the detection of failed links that are
frequently used. This is an intentional design choice as the goal of
\Sys is to identity high impact failures that affect many connections.

\noindent{\textbf{Finding the cause of other problems.}} \Sys's goal is to identify the cause of every packet drop, but other problems may also be of interest. \Sys can be extended to identify the cause of many such problems. For example, for latency, 
ETW provides TCP's smooth RTT estimates upon each received ACK. Thresholding on these values allows for identifying ``failed'' flows and \Sys's 
voting scheme can be used to provides a ranked list of suspects. Proving the accuracy of \Sys for such problems requires an extension of the analysis presented in this paper.

\noindent{\textbf{VM traffic problems.}} \Sys's
goal is to find the cause of drops on infrastructure
connections and through those, find the failed links in the
network. In principle, we can build a \Sys-like system to diagnose TCP failures for connections established by customer VMs as well. For example, we can update the monitoring agent to capture VM TCP statistics through a VFP-like system~\cite{vfp}.
However, such a system raises a number of new issues, chief among them being security. This is part of our future work.

 In conclusion, we stress that the purpose of \Sys is to {\em explain} the cause of drops when they occur. Many of these are not actionable and do not require operator intervention. The tally of votes on a given link provide a starting point for deciding when such intervention is needed.


\vspace{-1mm}
\section{Related Work}
\label{sec:relatedwork}
\vspace{-1mm}
Finding the source of failures in distributed systems,
specifically networks, is a mature topic. We outline
some of the key differences of \Sys with these works.

The most closely related work to ours is perhaps~\cite{alex},
which requires modifications to routers and both endpoints a limitation that \Sys does not have. Often services (e.g. storage) are unwilling to incur the additional overhead of new monitoring software on their machines and in many instances the two endpoints are in seperate organizations~\cite{Netpoirot}. Moreover, in order to
apply their approach to our datacenter, a number of engineering
problems need to be overcome, including finding a substitute
for their use of the DSCP bit, which is used for other purposes in our datacenter. Lastly, while the statistical testing method used in~\cite{alex} (as well as others) are useful when paths of both failed and non-failed flows are available they cannot be used in our setting as 
the limited number of traceroutes \Sys can send prevent it from tracking the path of all flows. In addition \Sys allows for diagnosis of individual connections {\em and} it works well in the presence of multiple simultaneous failures, features that~\cite{alex} does not provide. Indeed, finding paths only when they are needed is one of the most attractive features of \Sys as it minimizes its overhead on the system. Maximum cover algorithms~\cite{maxcover,pathdump} suffer from many of the same limitations described earlier for the binary optimization, since MAX COVERAGE and Tomo are approximations of~\eqref{opt:binary}.  Other related work can be loosely categorized as follows:

\noindent{\bf Inference and Trace-Based
  Algorithms~\cite{sherlock,netmedic,everflow,pingmesh,rinc,tulip,alex,lossradar,flowradar,verification}}
use anomaly detection and trace-based algorithms to find
sources of failures. They require knowledge/inference of
the location of logical devices, e.g. load balancers in
the connection path. While this information is available to the
network operators, it is not clear which instance of these entities a
flow will go over. This
reduces the accuracy of the results.
 
Everflow ~\cite{everflow} aims to accurately identify the path of packets of
interest. However, it does not scale to be used as an always on
diagnosis system. Furthermore, it requires additional features to be enabled in the switch. Similarly,~\cite{rachit,pathdump} provides
another means of path discovery, however, such approaches require deploying new applications to the remote 
end points which we want to avoid (due to reasons described in~\cite{Netpoirot}). Also, they depend on SDN enabled networks and are not applicable to our setting
where routing is based on BGP enabled switches.
  
Some inference approaches aim at \textit{covering} the full topology,
e.g.~\cite{pingmesh}. While this is useful, they typically
only provides a sampled view of connection livelihood and do not achieve the type of
always on monitoring that \Sys provides. The time between probes
for~\cite{pingmesh} for example is currently 5 minutes. It
is likely that failures that happen at finer time scales slip
through the cracks of its monitoring probes.

Other such work, e.g.~\cite{tulip,alex,lossradar,flowradar} require
access to both endpoints and/or switches.  Such access may not always
be possible. Finally, NetPoirot ~\cite{Netpoirot} can only identify
the general type of a problem (network, client, server) rather than
the responsible device.


\noindent\textbf{Network tomography~\cite{tomo1, tomo2, tomo3,
  tomo4, risk,shrink,tomo5,tomo6,tomo7,tomo8,duffield2,gestalt,geoff}}
typically consist of two aspects: (i) the gathering and filtering of
network traffic data to be used for identifying the points of
failure~\cite{tomo5,tomo1} and (ii) using the information found in the
previous step to identify where/why failures occurred~\cite{tomo2,
  tomo3,tomo4,gestalt,netdiagnoser,duffield2,newtomo}. \Sys utilizes
ongoing traffic to detect problems, unlike these approaches
which require a much heavier-weight operation of gathering large
volumes of data. Tomography-based approaches are
also better suited for non-transient failures, while \Sys can handle
both transient and persistent errors. \Sys also has coverage that
extends to the entire network infrastructure, and does not limit
coverage to only paths between designated monitors as some such approaches do. Work on analyzing
failures~\cite{tomo1,tomo4,tomo5,gestalt} are complementary and can be
applied to \Sys to improve our accuracy.






%

\noindent{\bf Anomaly detection~\cite{anom1,anom3,
    assaf,anom4,anomalbert,crovella2014,anomallydetection}} find
when a failure has occurred using machine learning~\cite{assaf,anom1}
and Fourier transforms~\cite{anomalbert}. \Sys goes a step further by finding the device
responsible.


\noindent{\bf Fault Localization by Consensus~\cite{netprofiler}}
assumes that a failure on a node common to the path used by a subset
of clients will result in failures on a significant number of
them. NetPoirot~\cite{Netpoirot} illustrates why this
approach fails in the face of a subset of problems that are common to
datacenters. While our work builds on this idea, it provides
a confidence measure that identifies how reliable a diagnosis report
is.


\noindent{\bf Fault Localization using TCP
  statistics~\cite{tcp1,tcp2,trat,snap,alex}} use TCP metrics for
diagnosis. \cite{tcp1} requires heavyweight active
probing. \cite{tcp2} uses learning techniques. Both~\cite{tcp2}, and
T-Rat~\cite{trat} rely on continuous packet captures which doesn't scale. SNAP~\cite{snap} identifies
performance problems/causes for connections by acquiring TCP
information which are gathered by querying socket options. It also
gathers routing data combined with topology data to compare the TCP
statistics for flows that share the same host, link, ToR, or aggregator switch. Given their lack of
continuous monitoring, all of these approaches fail in detecting the
type of problems \Sys is designed to detect. Furthermore, the goal of
\Sys is more ambitious, namely to find the link that causes packet
drops for each TCP connection.



\noindent{\bf Learning Based
  Approaches~\cite{dtree1,conext15,dtree2,Netpoirot}} do failure
detection in home and mobile networks. Our application domain is
different.

\noindent{\bf Application diagnosis~\cite{app1,mogul}} aim at
identifying the cause of problems in a distributed application's
execution path. The
limitations of diagnosing network level paths and the complexities
associated with this task are different. Obtaining all execution paths
{\em seen} by an application, is plausible in such systems but is not
an option in ours.

\noindent{\bf{Failure resilience in datacenters~\cite{f10,ankit,conga,mptcp,datacenterresilience,datacenter2,fattire,datacenter3,bodik,olaf}}}
target resilience to failures in datacenters. \Sys can be
helpful to a number of these algorithms as it can find
problematic areas which these tools can then help
avoid.


\noindent{\bf{Understanding datacenter failures~\cite{RAIL,gill2011}}}
aims to find the various types of failures in datacenters. They are useful in understanding the types of
problems that arise in practice and to ensure that our
diagnosis engines are well equipped to find them. 
\Sys's analysis agent uses the findings of~\cite{RAIL}.

\vspace{-2.5mm}


\pagebreak
\section{Conclusion}
\vspace{-1mm}
We introduced \Sys, an always on and scalable
monitoring/diagnosis system for datacenters. \Sys can accurately
identify drop rates as low as $0.05\%$ in datacenters with thousands
of links through monitoring the status of ongoing TCP flows. 
\vspace{-2mm}
\section{Acknowledgements}
\vspace{-2mm}
This work was was supported by grants NSF CNS-1513679, DARPA/I2O HR0011-15-C-0098. The authors would like to thank T. Adams, D. Dhariwal, A. Aditya, M. Ghobadi, O. Alipourfard, A. Haeberlen, J. Cao, I. Menache, S. Saroiu, and our shepherd H. Madhyastha for their help.

{\footnotesize \bibliographystyle{acm}
\bibliography{biblio}}

\newpage

\appendix
\label{sec:appendix}


\section{Application example: VM reboots}

In the introduction~(\S~1), we describe an instance in which the failure detection capacities of \Sys can be useful: pinpointing the cause of VM reboots. Indeed, in our datacenters, VM images are stored in a storage service. When a customer boots a VM, the image is mounted over the network. Thus, even a small network outage can cause the host kernel to ``panic'' and reboot the guest VM. We mentioned that over $70\%$ of VM reboots caused by network issues in our datacenters cannot be explained using currently deployed monitoring systems. To further illustrate how important this issue can be, Figure~\ref{fig:E17Dist} shows the number of unexplained VM reboots due to network problems in one day of operations: there were on average~$10$ VM reboots per hour due to unexplained network problems.

\begin{figure}[hbt]
  \centering
    \includesvg{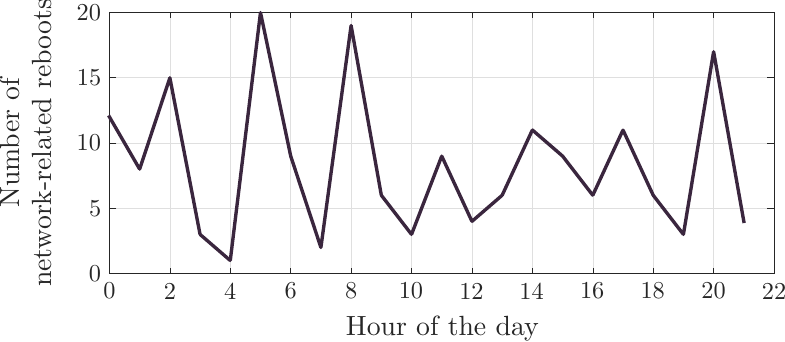}
     \caption{Number of network related reboots in a day.}
     \label{fig:E17Dist}
     \vspace{-5mm}
\end{figure}


\section{Network tomography example}

Knowing the path of all flows, it is possible to find with confidence which link dropped a packet. To do so, consider the example network in Figure~\ref{fig:toy}. Suppose that the link between nodes 2 and 4 drops packets. Flows 1--2 and 3--2 suffer from drops, but 1--3 does not. A set cover optimization, such as the one used by MAX COVERAGE and Tomo~\cite{netdiagnoser, risk}, that minimizes the number of ``blamed'' links will correctly find the cause of drops. This problem is however equivalent to a set covering optimization problem that is known to be NP-complete~\cite{combinatorial}.

\begin{figure}[t]
  \centering
    \includesvg[scale=1.2]{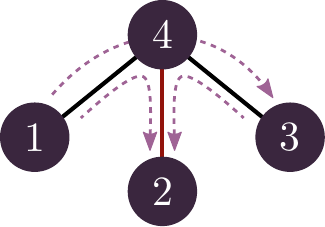}
    \caption{Simple tomography example\label{fig:toy}.}
	\vspace{-5mm}
\end{figure}

\section{Proofs}

\begin{table*}
\centering
\caption{Notation and nomenclature}
	\label{T:notation}
\begin{tabular}{cc}
\hline
$n_\text{pod}$ & Number of pods \\
\hline
$n_0$ & Number of top of the rack~(ToR) switches per pod \\
\hline
$n_1$ & Number of tier-1 switches per pod \\
\hline
$n_2$ & Number of tier-2 switches \\
\hline
Level~1 link & Link between ToR and tier-1 switch \\
\hline
Level~2 link & Link between tier-1 and tier-2 switch \\
\hline
$\calT_0^s$ & Set of all ToR switches in pod $s$ \\
\hline
$\calT_1^s$ & Set of all tier-1 switches in pod $s$ \\
\hline
$\calT_0$ & Set of all ToR switches~($\calT_0 = \calT_0^1 \cup \dots \cup \calT_0^{\npod}$) \\
\hline
$\calT_1$ & Set of all tier-1 switches~($\calT_1 = \calT_1^1 \cup \dots \cup \calT_1^{\npod}$) \\
\hline
$\calT_2$ & Set of all tier-2 switches \\
\hline
$k$ & Number of failed links in the network \\
\hline
$c_u$ & Upper bound on the number of packets per connection \\
\hline
$c_l$ & Lower bound on the number of packets per connection \\
\hline
$p_g$ & Probability that a good link drops a packet \\
\hline
$p_b$ & Probability that a failed link drops a packet \\
\hline
$v_g$ & Probability that a good link receives a vote \\
\hline
$v_b$ & Probability that a bad link receives a vote \\
\hline
$r_g$ & Probability that a good link causes a retransmission~(drops at least one packet) \\
\hline
$r_b$ & Probability that a bad link causes a retransmission~(drops at least one packet) \\
\hline
\end{tabular}
\end{table*}

\begin{definition}[Clos topology]

A Clos topology has $\npod$~pods each with~$n_0$ top of the rack~(ToR) switches under which lie $H$~hosts. The ToR switches are connected to~$n_1$ tier-1 switches by a complete network~($n_0 n_1$ links). Links between tier-0 and tier-1 switches are referred to as \emph{level~1 links}. The tier-1 switches within each pod are connected to~$n_2$ tier-2 switches by another complete network~($n_1 n_2$ links). Links between these switches are called \emph{level~2 links}. This notation is illustrated in Figure~\ref{F:Clos}.

\end{definition}

\begin{remark}[Communication and failure model]
	\label{R:Model}

Assume that connection occur uniformly at random between hosts under different ToR switches. Since the number of hosts under each ToR switch is the same, this is equivalent to saying that connections occur uniformly at random directly between ToR switches. Also, assume that link failure and connection routing are independent and that links drop packets independently across links and across packets.

\end{remark}

\begin{remark}[Notation]

We use calligraphic letter~($\calA$) to denote sets and boldface font~($\bA$) to denote random variables. Also, we write~$[M]$ to mean the set of integers between~$1$ and~$M$, i.e., $[M] = 1,\dots,M$.

\end{remark}

\subsection{Proof of Theorem~1}

%
%

\begin{proof}

Start by noticing that the number of hosts below each ToR switch is the same, so that we can consider that traceroute are sent on flows uniformly at random between ToR switches at a rate~$C_t H$. Moreover, note that routing probabilities are the same for links on the same level, so that the traceroute rate depends only on whether the link is on level~1 or level~2.

Since the probability of a switch routing a connection through any link is uniform, the traceroute rate of a level~1 link is given by
\begin{equation}\label{E:tracerouteBound1}
	R_{1} = \frac{1}{n_1} C_t H
		\text{,}
\end{equation}
Similarly for a level~2 link:
\begin{equation}\label{E:tracerouteBound2}
	R_{2} =  \frac{n_0}{n_1 n_2}
		\frac{n_0 (\npod - 1)}{(n_0 \npod - 1)} C_t H
		\text{,}
\end{equation}
where the second fraction represents the probability of a host connecting to another host outside its own pod, i.e., of going through a level~2 link. Since~$n_0$ links are connected to a tier-1 switch and~$n_1$ links are connected to a tier-2, the rate of ICMP packets at any links is bounded by~$T \leq \max\left[ n_0 R_1, n_1 R_2 \right]$.
Taking $\max\left[ n_0 R_1, n_1 R_2 \right] \leq T_\text{max}$ yields~\eqref{eq:maxC}.
\end{proof}


\subsection{Proof of Theorem~2}

We prove the following more precise statement of Theorem~2.

\begin{theorem}
	\label{T:vigilWorks2}

In a Clos topology with~$n_0 \geq n_2$ and
$\npod \geq 1 + \max \left[ \frac{n_0}{n_1}, \frac{n_2(n_0 - 1)}{n_0(n_0 - n_2)}, 1 \right]$,
\Sys will rank with probability~$(1 - \epsilon)$ the~$k < \frac{n_2 (n_0 \npod - 1)}{n_0 (\npod - 1)}$ bad links that drop packets with probability~$p_b$ above all good links that drop packets with probability~$p_g$ as long as
\begin{equation}\label{E:packetCondition}
	p_g \leq \frac{1 - (1-p_b)^{c_l}}{\alpha c_u}
		\text{,}
\end{equation}
where~$c_l$ and~$c_u$ are lower and upper bounds, respectively, on the number of packets per connection,
\begin{equation}\label{E:alpha}
	\alpha =
    \frac{
        n_0 (4 n_0 - k) (\npod - 1)
    }{
        n_2 (n_0 \npod - 1) - n_0 (\npod - 1) k
    }
		\text{,}
\end{equation}
and
\begin{equation}\label{E:epsilon}
\begin{aligned}
	\epsilon &\leq
	e^{-N \DKL((1+\delta) v_g \| v_g)} +
		e^{-N \DKL((1-\delta) v_b \| v_b)}
	\\
		{}&= 2 e^{-\calO(N)}
		\text{,}
\end{aligned}
\end{equation}
with $v_g$ and $v_b$ being the probabilities of a good and bad link receiving a vote, respectively, $N$ being the total number of connections between hosts, and~$\DKL(q \| r )$ denoting the Kullback-Leibler divergence between two Bernoulli distributions with probabilities of success~$q$ and~$r$.

\end{theorem}

Before proceeding, note that the typical scenario in which~$n_0 \geq 2 n_2$ and~$\frac{n_2(n_0 - 1)}{n_0(n_0 - n_2)} \leq 1$, as in our data center, the condition on the number of pods from Theorem~\ref{T:vigilWorks2} reduces to~$\npod \geq 1 + \frac{n_0}{n_1}$.

\begin{proof}

The proof proceeds as follows. First, we show that if a link has higher probability of receiving a vote, then it receives more votes if a large enough number of connections~($N$) are established. We do so using large deviation theory~\cite{ldp}, so that we can show that this does not happen actually decreases exponentially in~$N$.

\begin{lemma}\label{T:whp}

If~$v_b \geq v_g$, \Sys will rank bad links above good links with probability~$(1-\epsilon)$ for~$\epsilon$ as in~\eqref{E:epsilon}.

\end{lemma}

With Lemma~\ref{T:whp} in hands, we then need to relate the probabilities of a link receiving a vote~($v_b,v_g$) to the link drop rates~($p_b,p_g$). This will allow us to derive the signal-to-noise ratio condition in~\eqref{E:packetCondition}. Note that the probability of a link receiving a vote is the probability of a flow going through the link \emph{and} that a retransmission occurs~(i.e., some link in the flow's path drops at least one packet). Hence, we relate these probabilities by exploiting the combinatorial structure of ECMP in the Clos topology.



\begin{lemma}\label{T:scoringProb}

In a Clos topology with~$n_0 \geq n_2$ and
$\npod \geq 1 + \max \left[ \frac{n_0}{n_1}, \frac{n_2(n_0 - 1)}{n_0(n_0 - n_2)}, 1 \right]$, it holds that for~$k \leq n_0$ bad links
\begin{subequations}\label{E:scoringProb}
\begin{align}
	v_b &\geq \frac{r_b}{n_0 n_1 \npod}
		\label{E:badScoringProb}
	\\
	v_g &\leq \frac{1}{n_1 n_2 \npod} \frac{n_0(\npod -1)}{n_0\npod -1}
	\left[
		(4 - \frac{k}{n_0}) r_g + \frac{k}{n_0} r_b
	\right]
		\label{E:goodScoringProb}
\end{align}
\end{subequations}
where~$r_b$ and~$r_g$ are the probabilities of a retransmission occurring due to a bad and a good link, respectively.

\end{lemma}

Before proving these lemmata, let us see how they imply Theorem~\ref{T:vigilWorks2}. From the~\eqref{E:scoringProb} in Lemma~\ref{T:scoringProb}, it holds that
\begin{equation}\label{E:retransmissionBound}
	r_b
	\geq
	\underbrace{
		\frac{
			n_0 (4 n_0 - k)(\npod -1)
		}{
			n_2(n_0\npod - 1) - n_0 (\npod -1) k
		}
	}_{\alpha} r_g
	\Rightarrow
	v_b \geq v_g
		\text{,}
\end{equation}
for~$k < \frac{n_2 (n_0 \npod-1)}{n_0 (\npod - 1)} < n_0$. Thus, in a Clos topology, if the probability of retransmission due to a bad link is large enough compare to a good link, i.e.,~$r_b \geq \alpha r_g$ for $\alpha$ as in~\eqref{E:alpha}, then we have that the probability of a bad link receiving a vote is larger than that of a good link~($v_b \geq v_g$).

Still, \eqref{E:retransmissionBound} gives a relation in terms of the probabilities of retransmission~($r_g,r_b$) instead of the packet drop rates~($p_g$,~$p_b$) as in~\eqref{E:packetCondition}. To obtain~\eqref{E:packetCondition}, note that the probability~$r$ of retransmission during a connection with $c$~packets due to a link that drops packets with probability~$p$ is~$r = 1 - (1-p)^c$. Since~$r$ is monotonically increasing in~$c$, we have that~$r_b \geq 1 - (1-p_b)^{c_l}$. Similarly, $r_g \leq 1 - (1-p_g)^{c_u}$. Using the fact~$(1-x)^n \geq 1 - nx$ yields~\eqref{E:packetCondition}.
\end{proof}

We now proceed with the proofs of Lemmata~\ref{T:whp} and~\ref{T:scoringProb}.

\begin{proof}[Proof of Lemma~\ref{T:whp}]

We start by noting that in a datacenter-sized Clos network, almost every connection has a hop count of~$5$. In our datacenter, this happens to~$97.5\%$ of connections. Therefore, we can approximate links votes by assuming all bad votes have the same value. Thus, suffices to determine how many votes each link has.

Since links cause retransmissions independently across connections~(see Remark~\ref{R:Model}), the number of votes received by a bad link is a binomial random variable~$\bB$ with parameters~$N$, the total number of connections, and~$v_b$, the probability of a bad link receiving a vote. Similarly, let~$\bG$ be the number of votes on a good link, a binomial random variable with parameters~$N$ and~$v_g$. \Sys will correctly rank the bad links if~$\bB \geq \bG$, i.e., when bad links receive more votes than good links. This event contains the event~$\calD = \{\bG \leq (1+\delta) N v_g \cap \bB \geq (1-\delta) N v_b\}$ for $\delta \leq \frac{v_b - v_g}{v_b + v_g}$. Using the union bound~$\Pr\left[ \bigcup_i \calE_i \right] \leq \sum_i \Pr\left[ \calE_i \right]$~\cite{probBook}, the probability of \Sys identifying the correct links is therefore bounded by
\begin{equation}\label{E:pBound}
\begin{aligned}
	\Pr(\bB \geq \bG) &\geq
	\Pr\left[ \bG \leq (1+\delta) N v_g \cap
				\bB \geq (1-\delta) N v_b \right]
	\\
	{}&\geq 1 - \Pr\left[ \bG \geq (1+\delta) N v_g \right]
	\\
	&\hphantom{{}\geq 1}
		{}-	\Pr\left[ \bB \leq (1-\delta) N v_b \right]
\end{aligned}
\end{equation}

To proceed, note that the probabilities in~\eqref{E:pBound} can be bounded using the large deviation principle~\cite{ldp}. Indeed, let~$\bS$ be a binomial random variable with parameters~$M$ and~$q$. For~$\delta > 0$ it holds that
\begin{subequations}\label{E:LDP}
\begin{align}
	\Pr\left[ \bS \geq (1+\delta) q M \right] &\leq
		e^{-M \DKL((1+\delta) q \| q)}
	\\
	\Pr\left[ \bS \leq (1-\delta) q M \right] &\leq
		e^{-M \DKL((1-\delta) q \| q)}
\end{align}
\end{subequations}
where $\DKL(q \| r )$ is the Kullback-Leibler divergence between two Bernoulli distributions with probabilities of success~$q$ and~$r$~\cite{infotheory}. Explicitly,
\begin{equation*}
	\DKL(q \| r ) =
		q \log \left( \frac{q}{r} \right) +
		(1 - q) \log \left( \frac{1 - q}{1 - r} \right)
		\text{.}
\end{equation*}
Substituting the inequalities~\eqref{E:LDP} into~\eqref{E:pBound} yields~\eqref{E:epsilon}.
\end{proof}

\begin{figure}[tb]
	\centering
	\includesvg{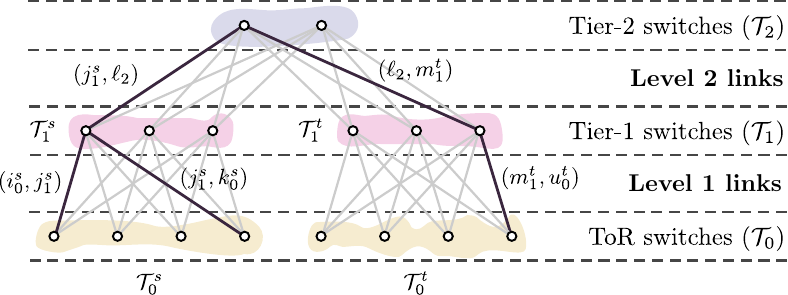}
	\caption{Illustration of notation for Clos topology used in the proof of Lemma~\ref{T:scoringProb}}
	\label{F:Clos}
\end{figure}

\begin{proof}[Proof of Lemma~\ref{T:scoringProb}]

Before proceeding, let~$\calT_0$, $\calT_1$, and~$\calT_2$ denote the set of ToR, tier-1, and tier-2 switches respectively~(Figure~\ref{F:Clos}). Also let~$\calT_0^s$ and~$\calT_1^s$, $s = [\npod]$, denote the tier-0 and tier-1 switches in pod~$s$ respectively. Note that~$\calT_0 = \calT_0^1 \cup \dots \cup \calT_0^{\npod}$ and~$\calT_1 = \calT_1^1 \cup \dots \cup \calT_1^{\npod}$. Note that we use subscripts to denote the switch tier and superscripts to denote its pod. To clarify the derivations, we maintain this notation for indices. For instance, $i_0^s$ is the $i$-th tier-0 switch from pod~$s$, i.e., $i_0^s \in \calT_0^s$, and~$\ell_2$ is the $\ell$-th tier-2 switch. Note that tier-2 switches do not belong to specific pods. We write~$(i_0^s,j_1^s)$ to denote the level~1 link that connects~$i_0^s$ to~$j_1^s$~(as in Figure~\ref{F:Clos}) and use~$r(i_0^s,j_1^s) = r(j_1^s,i_0^s)$ to refer to the probability of link~$(i_0^s,j_1^s)$ causing a retransmission. Note that $r$ is also a function of the number of packets in a connection, but we omit this dependence for clarity.

\begin{subequations}\label{E:level1ScoringProb}

The bounds in~\eqref{E:scoringProb} are obtained by decomposing the events that \Sys votes for a level~1 or level~2 link into a union of simpler events. Before proceeding, note that each connection only goes through one link in each level and in each direction, so that events such as ``going through a ToR to tier-1 link'' are disjoint.

Starting with level~1, let~$\bA_0$ be the event that a connection goes through link~$(i_0^s,j_1^s)$, i.e., a link that connects a ToR to a tier-1 switch in any pod. This event happens with probability
\begin{equation}
	\Pr\left[ \bm{A_0} \right] = \frac{1}{n_0 n_1 \npod}
		\text{,}
\end{equation}
given that there are~$n_0 n_1 \npod$ level~1 links and that connections occur uniformly at random. The link~$(i_0^s,j_1^s)$ will get a vote if one of five things occur: (i)~it causes a retransmission; (ii)~the connection stays within the pod and some other link causes a retransmission; (iii)~the connection leaves the pod and a link between a tier-1 and tier-2 switch causes the retransmission; (iv)~the connection leaves the pod and a link between a tier-2 and tier-1 switch causes the retransmission; or (v)~the connection leaves the pod and a link between a tier-1 and ToR switch in the other pod causes the retransmission. Formally, the link~$(i_0^s,j_1^s)$ receives a vote if a connection goes through it~(event $\bA_0$) \emph{and} either of the following occurs:
\begin{itemize}

\item event $\bA_1$: $(i_0^s,j_1^s)$ causes a retransmission, i.e.,
\begin{equation}
	\Pr\left[ \bm{A_1} \right] = r(i_0^s,j_1^s)
\end{equation}

\item event $\bA_2$: the connection also goes through some $(j_1^s,k_0^s)$, $k_0^s \neq i_0^s$, and~$(j_1^s,k_0^s)$ causes a retransmission. Therefore,
\begin{equation}
	\Pr\left[ \bA_2 \right] =
	\underbrace{
		\frac{1}{n_0 \npod - 1}
	}_\text{connect to $k_0^s$}
	\sum_{k_0^s \in \calT_0^s \setminus \{i_0^s\}} r(j_1^s,k_0^s)
\end{equation}

\item event $\bA_3$: the connection also goes through some $(j_1^s,\ell_2)$ and~$(j_1^s,\ell_2)$ causes a retransmission, which occurs with probability
\begin{equation}
	\Pr\left[ \bA_3 \right] =
	\underbrace{
		\frac{n_0 (\npod - 1)}{n_0 \npod - 1}
	}_\text{leave pod $s$}
	\underbrace{
		\frac{1}{n_2}
	}_{\substack{\text{go through}\\\ell_2}}
	\sum_{\ell_2 \in \calT_2} r(j_1^s,\ell_2)
\end{equation}

\item event $\bA_4$: the connection also goes through some $(\ell_2,m_1^t)$, $t \neq s$, and~$(\ell_2,m_1^t)$ causes a retransmission, so that
\begin{equation}
	\Pr\left[ \bA_4 \right] =
	\underbrace{
		\frac{n_0}{n_0 \npod - 1}
	}_\text{go to pod $t$}
	\underbrace{
		\frac{1}{n_1 n_2}
	}_{\substack{\text{go through}\\(\ell_2,m_1^t)}}
	\sum_{\substack{
		\ell_2 \in \calT_2,
		\\
		m_1^t \in \calT_1^t,
		\\
		t \in [\npod] \setminus s}}
	r(\ell_2,m_1^t)
\end{equation}

\item event $\bA_5$: the connection also goes through some $(m_1^t,u_0^t)$, $t \neq s$, and~$(m_1^t,u_0^t)$ causes a retransmission. Thus,
\begin{equation}
	\Pr\left[ \bA_5 \right] =
	\underbrace{
		\frac{1}{n_0 \npod - 1}
	}_\text{go to pod $t$}
	\underbrace{
		\frac{1}{n_1}
	}_{\substack{\text{go through}\\m_1^t}}
	\sum_{\substack{
		m_1^t \in \calT_1^t,
		\\
		u_0^t \in \calT_0^t,
		\\
		t \in [\npod] \setminus s}}
	r(m_1^t,u_0^t)
\end{equation}

\end{itemize}

\end{subequations}

\begin{subequations}\label{E:level2ScoringProb}

Similarly for level~2, let~$\bB_0$ be the event that a connection goes through link~$(j_1^s, \ell_2)$, so that its probability is
\begin{equation}
	\Pr\left[ \bB_0 \right] =
	\underbrace{
		\frac{1}{\npod}
	}_\text{start in pod $s$}
	\underbrace{
		\frac{n_0 (\npod - 1)}{n_0 \npod - 1}
	}_\text{leave pod $s$}
	\underbrace{
		\frac{1}{n_1 n_2}
	}_{\substack{\text{go through}\\(j_1^s, \ell_2)}}
\end{equation}
For this link to receive a vote either (i)~it causes a retransmission; (ii)~a level~1 link from the origin pod causes a retransmission; (iii)~a link between a tier-2 and tier-1 switch causes the retransmission; or (iv)~a level~1 link in the destination pod causes the retransmission. Thus, link~$(j_1^s, \ell_2)$ gets a vote if a connection goes through $(j_1^s, \ell_2)$~(event~$\bB_0$) \emph{and} either of the following occurs:
\begin{itemize}

\item event $\bB_1$: $(j_1^s, \ell_2)$ causes a retransmission, i.e.,
\begin{equation}
	\Pr\left[ \bB_1 \right] = r(j_1^s, \ell_2)
\end{equation}

\item event $\bB_2$: the connection also goes through some $(i_0^s,j_1^s)$ and~$(i_0^s,j_1^s)$ causes a retransmission. Then,
\begin{equation}
	\Pr\left[ \bB_2 \right] =
	\underbrace{
		\frac{1}{n_0}
	}_\text{start in $i_0^s$}
	\sum_{i_0^s \in \calT_0^s} r(i_0^s,j_1^s)
\end{equation}

\item event $\bB_3$: the connection also goes through some $(\ell_2,m_1^t)$, $t \neq s$, and~$(\ell_2,m_1^t)$ causes a retransmission, which yields
\begin{equation}
	\Pr\left[ \bB_3 \right] =
	\underbrace{
		\frac{1}{n_1 (\npod - 1)}
	}_\text{go through $m_1^t$}
	\sum_{\substack{
		m_1^t \in \calT_1^t,
		\\
		t \in [\npod] \setminus s}}
	r(\ell_2,m_1^t)
\end{equation}

\item event $\bB_4$: the connection also goes through some $(m_1^t,u_0^t)$, $t \neq s$, and~$(m_1^t,u_0^t)$ causes a retransmission. Therefore,
\begin{equation}
	\Pr\left[ \bB_4 \right] =
	\underbrace{
		\frac{1}{n_0 n_1 (\npod - 1)}
	}_{\substack{\text{go through}\\(m_1^t,n_0^t),\ t \neq s}}
	\sum_{\substack{
		m_1^t \in \calT_1^t,
		\\
		u_0^t \in \calT_0^t,
		\\
		t \in [\npod] \setminus s}}
	r(m_1^t,u_0^t)
\end{equation}

\end{itemize}
\end{subequations}

To obtain the lower bound in~\eqref{E:badScoringProb}, note that a bad link receives at least as many votes as retransmissions it causes. Therefore, the probability of \Sys voting for a bad link is larger than the probability of that link causing a retransmission. Explicitly, using the fact that failure and routing are independent and~$r = r_b$, \eqref{E:level1ScoringProb} and~\eqref{E:level2ScoringProb} give
\begin{align*}
	v_b &\geq \min \left[
		\Pr(\bA_0 \cap \bA_1), \Pr(\bB_0 \cap \bB_1)
	\right]
	\\
	{}&= \min \left[
		\frac{1}{n_0 n_1 \npod},
		\frac{1}{n_1 n_2 \npod} \frac{n_0 (\npod - 1)}{n_0 \npod - 1}
	\right] r_b
		\text{.}
\end{align*}
The assumption that~$\npod \geq 1 + \frac{n_2 (n_0 - 1)}{n_0 (n_0 - n_2)}$ makes the first term smaller than the second and yields~\eqref{E:badScoringProb}.

In contrast, the upper bound in~\eqref{E:goodScoringProb} is obtained by applying the union bound~\cite{probBook} to~\eqref{E:level1ScoringProb} and~\eqref{E:level2ScoringProb}. Indeed, this leads to the following inequalities for the probability of \Sys voting for a good level~1 and level~2 link:
\begin{subequations}\label{E:levelsProbBound}
\begin{align}
	v_{g,1} &=
	\Pr\left[ \bA_0 \cap (\bA_1 \cup \bA_2 \cup
				\bA_3 \cup \bA_4 \cup \bA_5) \right]	
	\notag\\
	{}&\leq \Pr[ \bA_0 ]
		\left(
			\sum_{i = 1}^{5} \Pr[\bA_i]
		\right)
		\label{E:level1ProbBound}
	\\
	v_{g,2} &=
	\Pr\left[ \bB_0 \cap
				(\bB_1 \cup \bB_2 \cup \bB_3 \cup \bB_4) \right]
	\notag\\
	{}&\leq \Pr[ \bB_0 ]
		\left(
			\sum_{i = 1}^{4} \Pr[\bB_i]
		\right)
		\label{E:level2ProbBound}
\end{align}
\end{subequations}
where~$v_{g,1}$ and~$v_{g,2}$ denote the probability of a good level~1 and level~2 link being voted bad, respectively. Note that once again used the independence between failures and routing. From~\eqref{E:levelsProbBound}, it is straightforward to see that~$v_g \leq \max \left[ v_{g,1}, v_{g,2} \right]$.

To obtain~\eqref{E:goodScoringProb}, we first bound~\eqref{E:levelsProbBound} by assuming that all~$k$ bad links belong to the event~$\bA_i$ and~$\bB_i$, $i \geq 2$, that maximize~$v_{g,1}$ and~$v_{g,2}$. For a good level~1 link, it is straightforward to see from~\eqref{E:level1ScoringProb} that since~$n_0 \geq n_2$, event~$\bA_3$ has the largest coefficient. Thus, taking all links to be good except for~$k$ bad links satisfying~$\bA_3$ one has
\begin{multline}\label{E:boundSg1}
	v_{g,1} \leq \frac{1}{n_0 n_1 \npod}
		\frac{n_0 (\npod - 1)}{n_0 \npod - 1} \times{}
		\\
		\left[
			\left(
				4 - \frac{k}{n_2} +
				\frac{2 (n_0 - 1)}{n_0 (\npod - 1) }
			\right) r_g
			+
			\frac{k}{n_2} r_b
	\right]
	\text{,}
\end{multline}
which holds for $k \leq n_2$. Similarly for a good level~2 link, since~$\npod \geq \frac{n_0}{n_1} + 1$ it holds from~\eqref{E:level2ScoringProb} that event~$\bB_2$ has the largest coefficient. Therefore,
\begin{multline}\label{E:boundSg2}
	v_{g,2} \leq \frac{1}{n_1 n_2 \npod}
		\frac{n_0 (\npod - 1)}{n_0 \npod - 1} \times{}
		\\
		\left[
			\left( 4 - \frac{k}{n_0} \right) r_g
			+
			\frac{k}{n_0} r_b
		\right]
		\text{,}
\end{multline}
which holds for~$k \leq n_0$. Straightforward algebra shows that for~$\npod \geq 2$, $v_{g,2} \geq v_{g,1}$, from which~\eqref{E:badScoringProb} follows.
\end{proof}


\section{Greedy solution of the binary program}

\begin{algorithm}[t]
 \begin{algorithmic}[1]
\STATE $\calF$: set of failed links
\STATE $\calC$: set of failed connections
\STATE $\calF \gets \emptyset$
\WHILE{$\calC \neq \emptyset$}
  \STATE $l \gets$ link that explains the most number of additional failures
  \STATE $\calL \gets$ failures explained by~$l$
  \STATE $\calF \gets \calF \cup \{l\}$
  \STATE $\calC \gets \calC - \calL$
\ENDWHILE
\RETURN  $\calF$
\end{algorithmic}
	\caption{Finding the most problematic links in the network.}
	\label{L:greedy}
\end{algorithm}

When discussing optimization-based alternatives to \Sys's voting scheme, we presented the following problem which we dubbed the \emph{binary program}
\begin{equation}\label{opt:binary2}
\begin{aligned}
	\minimize&
	&&\norm{\bm{p}}_0
	\\
	\text{subject to}&
	&&\bm{A} \bm{p} \geq \bm{s}
	\\
	&&&\bm{p} \in \{0,1\}^L
\end{aligned}
\end{equation}
where~$\bm{A}$ is a~$C \times L$ routing matrix; $\bm{s}$ is a~$C \times 1$ vector that collects the status of each flow during an epoch (each element of~$\bm{s}$ is~$1$ if the connection experienced at least one retransmission and~$0$ otherwise); $L$ is the number of links; $C$ is the number of connections in an epoch; and~$\norm{\bm{p}}_0$ denotes the number of nonzero entries of the vector~$\bm{p}$.

Problem~\eqref{opt:binary2} can be described as one of looking for the \emph{smallest number of links that explains all failures}. To see this is the case, start by noting that~$\bm{p}$ is an~$L \times 1$ whose $i$-th entry describe whether link~$i$ is believed to have failed or not. Thus, since~$\bm{A}$ is the routing matrix, the~$C \times 1$ vector~$\bm{A} \bm{p}$ describes whether~$\bm{p}$ explains a possible failure in that connection or not: if~$[\bm{A} \bm{p}]_i = 0$, then~$\bm{p}$ does not explain a possible failure in the~$i$-th connection; if~$[\bm{A} \bm{p}]_i > 0$, then~$\bm{p}$ explains a possible failure in the~$i$-th connection. Hence, the constraint~$\bm{A} \bm{p} \geq \bm{s}$ can be read as ``explain each failure at least once''. Note that in an attempt to explain \emph{all} failure, $\bm{p}$ may explain failures that did not occur. This is the reason we used the term ``possible failure'' earlier. Recalling that the objective of~\eqref{opt:binary2} is to minimize the number of ones in~$\bm{p}$, i.e., the number of links marked as ``failed'', gives the interpretation from the beginning of the paragraph.

As we noted before, the binary program is NP-hard in general. Its solution is therefore typically approximated using a greedy procedure. Given the interpretation from the last paragraph, we can describe the greedy solution of~\eqref{opt:binary2} as in Algorithm~\ref{L:greedy}~\cite{combinatorial}. The algorithm proceeds as follows. Start with an empty set of failed links~$\calF$ and a set of unexplained failures~$\calC$. At each step, find the single link~$l$ that explains the largest number of unexplained failures, add it to~$\calF$, and remove from~$\calC$ all the failures it explains. We then iterate until~$\calC$ is empty. Note that this is the procedure followed by MAX COVERAGE and Tomo~\cite{netdiagnoser, risk}. They both therefore approximate the solution of~\eqref{opt:binary2}.

\end{document}